\documentclass[a4paper,12pt]{article}
\usepackage[utf8]{inputenc}
\usepackage[UKenglish]{isodate}
\cleanlookdateon
\usepackage[margin=1.25in]{geometry}
\usepackage[usenames,dvipsnames]{xcolor}
\usepackage{tikz,adjustbox}
\usetikzlibrary{matrix,arrows,calc,shapes,positioning,decorations.pathmorphing,patterns}
\usepackage{amssymb}
\usepackage{amsmath}
\usepackage{amsthm}
\usepackage{newpxtext,newpxmath}
\usepackage{csquotes}
\usepackage[hyphens]{url}
\usepackage{hyperref}
\usepackage[hyphenbreaks]{breakurl}
\usepackage{enumitem}
\usepackage{thm-restate}
\usepackage{esvect}
\usepackage{setspace}
\usepackage{titlesec}
\usepackage{bbm}
\usepackage{cleveref}
\usepackage{microtype}

\hypersetup{breaklinks, colorlinks,
  citecolor=orange!60!black}

\usepackage[authoryear, round]{natbib}
\newcommand{\ceil}[1]{\left \lceil #1 \right \rceil}
\newcommand{\abs}[1]{\left \vert #1 \right \vert}
\newcommand{\eps}{\varepsilon}
\newcommand{\mbf}[1]{\mathbf{#1}}
\newcommand{\N}{\mathbb{N}}
\newcommand{\Lnk}{\vv{L}(n, k)}

\newcommand{\cG}{\mathcal{G}}
\newcommand{\cL}{\mathcal{L}}

\newcommand{\cS}{\mathcal{S}}

\DeclareMathOperator{\Poi}{Poi}
\DeclareMathOperator{\Bin}{Bin}
\DeclareMathOperator{\Hypergeometric}{Hypergeometric}

\DeclareMathOperator{\Prob}{\mathbb{P}}
\DeclareMathOperator{\expec}{\mathbb{E}}

\newcommand\blfootnote[1]{
  \begingroup
  \renewcommand\thefootnote{}\footnote{#1}
  \addtocounter{footnote}{-1}
  \endgroup
}

  {\list{}{\leftmargin=0.15in\rightmargin=0.15in}\item[]}%
  {\endlist}

\theoremstyle{plain}
\newtheorem{theorem}{Theorem}
\newtheorem{corollary}[theorem]{Corollary}
\newtheorem{lemma}[theorem]{Lemma}
\newtheorem{proposition}[theorem]{Proposition}

\newtheorem{claim}{Claim}

\theoremstyle{definition}
\newtheorem{definition}[theorem]{Definition}
\theoremstyle{remark}

\newcommand{\zero}{0}
\newcommand{\one}{1}

\title{Game connectivity and adaptive dynamics in many-action games}
\author{Tom Johnston\thanks{School of Mathematics, University of Bristol, Bristol,      BS8 1UG, UK.}\:\,\thanks{Heilbronn Institute for Mathematical Research, Bristol, UK.}
    \quad Michael Savery\footnotemark[2]\:\,\thanks{Mathematical Institute, University of Oxford, Oxford, OX2 6GG, UK.}
    \quad Alex Scott\footnotemark[3]\:\,\thanks{Supported by EPSRC grant EP/X013642/1.}
    \quad Bassel Tarbush\thanks{Department of Economics, University of Oxford, Oxford, OX1 3UQ, UK.}}
\date{{\today}}

\advance\textwidth2\evensidemargin
\begin{document}

\maketitle

\begin{abstract}
    \noindent We study the typical structure of games in terms of their connectivity properties. A game is `connected' if it has a pure Nash equilibrium and there is a best-response path from every action profile which is not a pure Nash equilibrium to every pure Nash equilibrium; a game is generic if it has no indifferences. In previous work we showed that, among all $n$-player $k$-action generic games that admit a pure Nash equilibrium, the fraction that are connected tends to $1$ as $n$ gets sufficiently large relative to $k$. Here, we consider the large-$k$ regime, which behaves differently: we show that the connected fraction tends to $1-\zeta_n$ as $k$ gets large, where $\zeta_n>0$ is an explicit constant.  Thus, a constant fraction of many-action games are \emph{not} connected. However, for $n\geq3$, $\zeta_n$ is small and tends to $0$ rapidly with $n$, so as $n$ increases all but a vanishingly small fraction of many-player-many-action games are connected. Since connectedness is conducive to equilibrium convergence, we find a simple adaptive dynamic that is guaranteed to converge to a pure Nash equilibrium in all but a vanishingly small fraction of generic games that have one. We rely on new probabilistic and combinatorial arguments to tackle the large-$k$ regime.
  
    \blfootnote{\emph{Email}: \textsf{tom.johnston@bristol.ac.uk, \{savery,scott\}@maths.ox.ac.uk, bassel.tarbush@economics.ox.ac.uk}}
\end{abstract}

\section{Introduction}
The Nash equilibrium solution is a central concept in non-cooperative game theory. A justification that is sometimes given for expecting players to play such equilibria is that they can emerge as the outcomes of simple adaptive dynamics in which players independently employ simple decision heuristics in a decentralised way \citep*{holt2004nash,young2004strategic}. The qualification `simple adaptive' (which we formalise later) refers to the amount of information that each player has access to, and examples of such dynamics include, among others, the well-studied classes of better- and best-response dynamics \citep*{sandholm2010population}. However, an influential result due to \citet*{hart2003uncoupled,hart2006stochastic} has posed a significant challenge to such a justification: for any simple adaptive dynamic, one can construct a game with a pure Nash equilibrium on which the dynamic is not guaranteed to lead to a pure Nash equilibrium of that game. 

While \citet*{hart2003uncoupled,hart2006stochastic} have shown that simple adaptive dynamics perform (very) poorly in the worst case (i.e.\ in certain games), we show that the `typical' case is very different. In previous work, we considered the space of all $n$-player $k$-action games that are ordinal, generic (i.e.\ without indifferences), and have a pure Nash equilibrium, and we used probabilistic and combinatorial arguments to show that when $n$ is sufficiently large relative to $k$, all but a vanishingly small fraction of such games are `connected' \citep*{johnston2023game}; a game is said to be \emph{connected} if there is a best-response path from every action profile which is not a pure Nash equilibrium to every pure Nash equilibrium. The result shows that connectedness is a typical feature of generic games that have a pure Nash equilibrium and relatively many players, despite connectedness being a strong property. Since connectedness is conducive to the equilibrium convergence of simple adaptive dynamics, this allowed us to show that there is a simple adaptive dynamic that is guaranteed to converge to a pure Nash equilibrium in all but a vanishingly small fraction of generic $n$-player $k$-action games that have one, provided $n$ is sufficiently large relative to $k$. However, the techniques used in our previous work could not tackle the regime in which the number of actions per player $k$ is large relative to $n$, and that regime was therefore left open. The central technical contribution of our present paper is in developing techniques to address the many-actions regime.

Here, we show that among $n$-player generic games that admit a pure Nash equilibrium, the fraction that are connected tends to $1-\zeta_n$ as $k$ gets large, where $\zeta_n>0$ (\Cref{prop:large_k}). Since a constant fraction of many-action games are \emph{not} connected, this contrasts with the many-players (i.e.~large $n$) regime studied in our previous work.  However, for $n\geq 3$, $\zeta_n$ is small and tends to $0$ rapidly with $n$. By combining the results of this paper with those of \citet*{johnston2023game} we also show that, among generic games that admit a pure Nash equilibrium, the fraction that are connected tends to 1 as the number of players $n$ gets large, \emph{regardless} of the number of actions $k$ per player (\Cref{thm:inf_k}). Consequently, \Cref{prop:large_k} and \Cref{thm:inf_k} give us that the fraction is 1 in both iterated limits $\lim_{n \to \infty} \lim_{k \to \infty}$ and $\lim_{k \to \infty} \lim_{n \to \infty}$, and in the double limit $\lim_{(n,k) \to (\infty,\infty)}$.  We note, however, that \Cref{prop:large_k} and \Cref{thm:inf_k} neither imply nor are implied by the results of \citet*{johnston2023game}. This paper therefore complements our previous work to provide a more complete picture of connectedness in games. The combined results of this paper and our previous work establish that a simple adaptive dynamic can `typically' find pure Nash equilibria (\Cref{prop:dynamics}). Moreover, as noted in \citet*{johnston2023game}, results on game connectivity also have implications for the convergence of more complex dynamics that are not necessarily `simple adaptive' dynamics.

Games with many actions arise naturally in a number of contexts including AI alignment (where AI models engage in simultaneous debate in language-space, see e.g.~\citealp*{chen2024playinglargegamesoracles,collina2025emergentalignmentcompetition,buzaglo2025hiddengameproblem}), certain games with combinatorial action spaces (such as combinatorial auctions and Blotto games), and learning in discretized games \citep*{bichler2023computing,bichler2023learning}.\footnote{\cite*{bichler2023computing,bichler2023learning} employ simultaneous online dual averaging (SODA), a learning dynamic, to successfully compute Nash equilibria in auction settings in which no analytic solution for the equilibria is known. In their case, the original game has a continuous action space, but SODA is applied to a discretized version with a large number of actions per player.} There is also burgeoning theoretical work on better- and best-response dynamics in games with large action spaces \citep*{amiet2021better,mimun2024best,collevecchio2024basins,collevecchio2025findingnashequilibriumrandom,ashkenazigolan2025simultaneousbestresponsedynamicsrandom}.

\paragraph*{Roadmap} Our game-theoretic definitions, our results on game `connectedness', and their implications for adaptive dynamics are in Sections \ref{sec:model}, \ref{sec:connectedness}, and \ref{sec:dynamics}, respectively. A brief discussion of related literature is provided in Section \ref{sec:related}. Our technical results, on which our results for games are based, as well as their proofs are in Sections \ref{sec:random}-\ref{sec:goodsinks}.

\subsection{Model}\label{sec:model}
In this section, we recall some standard definitions from the theory of games and introduce our notation. For $n\in \N$, we use $[n]$ as shorthand for the set $\{1, \dots, n\}$. For each $a\in \N^n$ and $i\in [n]$, we write $a_{-i}$ for the element of $\N^{n-1}$ obtained by deleting the $i$th coordinate of $a$. In an abuse of notation, for $x\in \N$ and $a_{-i}\in \N^{n-1}$, we write $(x,a_{-i})$ for the element of $\N^n$ obtained by inserting $x$ into the $i$th coordinate of $a_{-i}$.

A \emph{game} is a tuple
\[
    \Big([n], \underbrace{\big([k],\dots,[k]\big)}_{\text{length $n$}} , (\succsim_i)_{i \in [n]}\Big),
\]
where $n$ and $k$ are integers, each no less than 2, and for each $i$, $\succsim_i$ is a complete and transitive binary relation on $[k]^n$. We say that $[n]$ is the \emph{player set} of the game and that each player has an \emph{action set} $[k]$.\footnote{Our approach can accommodate games with different numbers of actions across players but, for simplicity, we restrict ourselves to games in which all players have the same number of actions.} Elements of $[k]^n$ are called \emph{action profiles}, and $\succsim_i$ is known as $i$'s \emph{preference relation}. For each $i$, $\succ_i$ denotes the asymmetric part of $\succsim_i$.

An action $a_i$ of player $i$ is a \emph{best-response} to $a_{-i}$ if $(a_i, a_{-i}) \succsim_i (x, a_{-i})$ for every $x \in [k]$. An action profile $a \in A$ is a \emph{pure Nash equilibrium} if for each player $i \in [n]$, $a_i$ is a best-response to $a_{-i}$. The \emph{best-response graph} of an $n$-player $k$-action game is the directed graph $( [k]^n, \rightarrow )$ whose vertex set is the set of action profiles $[k]^n$ and whose directed edge set $\rightarrow$ is such that, for every $a,b\in [k]^n$, $a \rightarrow b$ if and only if there exists $i\in[n]$ such that $a_{-i}=b_{-i}$, the action $b_i$ is a best-response to $a_{-i}$, and $b \succ_i a$. In other words, there is a directed edge from $a$ to $b$ whenever $b_i$ is a strict best-response to $a_{-i}=b_{-i}$ for some player $i$. Figure \ref{fig:intro_example} shows a game and its associated best-response graph. Relatedly, we also consider a game's better-response graph. An action $a_i$ of player $i$ is a \emph{better-response} than $a_i'$ to $a_{-i}$ if $(a_i,a_{-i}) \succ_i (a_i',a_{-i})$. The \emph{better-response graph} of an $n$-player $k$-action game is the directed graph $( [k]^n, \rightarrow )$ whose vertex set is the set of action profiles $[k]^n$ and whose directed edge set $\rightarrow$ is such that, for every $a,b\in A$, $a \rightarrow b$ if and only if there exists $i\in[n]$ such that $a_{-i}=b_{-i}$ and $b_i$ is a better-response to $a_{-i}$ than $a_i$.

\begin{figure}
\centering
\begin{tikzpicture}[scale=0.75]
\begin{scope}[scale=0.8,xshift=-80,yshift=70,every node/.append style={yslant=0,xslant=0.8},yslant=0,xslant=0.8]
\draw[xstep=2cm,ystep=1cm,color=gray] (0,0) grid (4,2);
\node at (1,1.5) {\footnotesize $\one,\one,\one$}; \node at (3,1.5) {\footnotesize$\zero,\zero,\one$};
\node at (1,0.5) {\footnotesize$\zero,\one,\zero$}; \node at (3,0.5) {\footnotesize $\one,\zero,\zero$};
\node[] at (-1,0.5) {$\mathbf{B}$};
\node[] at (-1,1.5) {$\mathbf{A}$};
\node[] at (3,2.5) {$\mathbf{B}$};
\node[] at (1,2.5) {$\mathbf{A}$};
\node[] at (5,1) {$\mathbf{A}$};
\draw[white] (-1.5,0.3) -- (-1.5,2) node [pos=0.5,above,rotate=90,yshift=0.2cm] {\color{black}Player 1};
\draw[white] (0,3) -- (3.2,3) node [pos=0.5,above] {\color{black}Player 2};
\end{scope}

\begin{scope}[scale=0.8,xshift=-80,yshift=0,every node/.append style={yslant=0,xslant=0.8},yslant=0,xslant=0.8]
\draw[xstep=2cm,ystep=1cm,color=gray] (0,0) grid (4,2);
\node at (1,1.5) {\footnotesize $\one,\zero,\zero$};  \node at (3,1.5) {\footnotesize $\zero,\one,\zero$};
\node at (1,0.5) {\footnotesize $\zero,\zero,\one$}; \node at (3,0.5) {\footnotesize $\one,\one,\one$};
\node[] at (5,1) {$\mathbf{B}$};
\end{scope}

\draw[scale=0.8,white] (3.8,1) -- (3.8,3.75) node [pos=0.5,above,rotate=270] {\color{black}Player 3};

\begin{scope}[xshift=150,yshift=-20,
scale=3.5,
roundnode/.style={rectangle, draw=white, fill=white,inner sep=0,outer sep=0}]
\foreach \x in {1,2}
\foreach \y in {1,2}
\foreach \z in {1,2}
{\node[] (\z\x\y) at (\x,\y,\z) {$\circ$};} 
\node[] () at (1,2,1) {$\bullet$};
\node[] () at (2,1,2) {$\bullet$};
\path[->] (211) edge [thick ]  (111);
\path[->] (121) edge [thick ]  (221);
\path[->] (111) edge [thick ]  (121);
\path[->] (211) edge [thick ]  (221);

\path[->] (212) edge [thick ]  (112);
\path[->] (122) edge [thick ]  (222);
\path[->] (122) edge [thick ]  (112);
\path[->] (222) edge [thick ]  (212);

\path[->] (212) edge [thick ]  (211);
\path[->] (111) edge [thick ]  (112);
\path[->] (121) edge [thick ]  (122);
\path[->] (222) edge [thick ]  (221);
\end{scope}
\end{tikzpicture}
\caption{A 3-player 2-action game (left) and its corresponding best-response graph (right). The $\bullet$ vertices are sinks.  This figure is repeated from \citet*{johnston2023game}.}
\label{fig:intro_example}
\end{figure}
We now define various classes of games in terms of the connectivity properties of their best-response graphs. As part of these definitions we will use standard terminology from the theory of directed graphs which we briefly recall here. Given a directed graph $(V,\rightarrow)$ with vertex set $V$ and edge set $\rightarrow$, a vertex $v \in V$ is a \emph{sink} if it has no outgoing edges, and a \emph{non-sink} otherwise. Similarly, a vertex $v \in V$ is a \emph{source} if it has no incoming edges, and a \emph{non-source} otherwise. Observe that the pure Nash equilibria of a game correspond to the sinks of its best-response graph. For any pair of vertices $v,v' \in V$, we say that $v$ can \emph{reach} $v'$ if there is a sequence $(v^1,\dots,v^m)$ of vertices with $v^1 = v$ and $v^m = v'$ such that $v^i \rightarrow v^{i+1}$ for all $i\in[m-1]$; in this case we also say that the vertex~$v'$ can \emph{be reached from} $v$. Note that every vertex can reach and be reached from itself. A \emph{cycle} is a sequence $(v^1,\dots,v^m)$ of distinct vertices that has length $m \geq 2$ and that satisfies $v^m\rightarrow v^1$ and $v^i \rightarrow v^{i+1}$ for all $i\in[m-1]$. We say that a game is \emph{acyclic} if its best-response graph has no cycles, and we say that a game is \emph{weakly acyclic} if its best-response graph has the property that every vertex can reach a sink. This terminology has become standard in the literature. Interest in such connectivity properties stems from the fact that they are a key determinant of the behaviour of adaptive dynamics. Weak acyclicity, for example, is a necessary condition for the convergence of best-response dynamics to a pure Nash equilibrium from any initial action profile \citep*{young1993evolution,fabrikant2013structure,jaggard2014self,apt2015classification}. The well-known class of potential games is a subclass of acyclic games, on which many adaptive dynamics are guaranteed to converge to equilibrium \citep*{monderer1996potential,roughgarden2016twenty}. In addition to acyclicity and weak acyclicity, we say that a game is \emph{connected} if its best-response graph has at least one sink and the property that every non-sink can reach every sink. This property was recently introduced in \citet*{johnston2023game}. Finally, for each connectivity property $P \in \{$acyclic, weakly acyclic, connected$\}$, we say that a game is \emph{better-response} $P$ if its better-response graph has that property. For example, a game is better-response connected if its better-response graph has a sink and the property that every non-sink can reach every sink. We obtain the following logical relationships between the game classes.\footnote{Note that connectedness neither implies nor is implied by acyclicity.}

\begin{adjustbox}{width=\linewidth}
\centering
    \begin{tikzpicture}
        \matrix (m) [matrix of math nodes,row sep=1em,column sep=1em,minimum width=0em] {
            \text{better-response acyclic} & \text{acyclic} & \text{weakly acyclic} & \text{better-response weakly acyclic}\\
                                    &                & \text{connected}      & \text{better-response connected}\\
        };
        \path[-stealth]
        (m-1-1) edge[double] (m-1-2)
        (m-1-2) edge[double] (m-1-3)
        (m-1-3) edge[double] (m-1-4)
        (m-2-3) edge[double] (m-2-4)
        (m-2-3) edge[double] (m-1-3)
        (m-2-4) edge[double] (m-1-4);
    \end{tikzpicture} 
\end{adjustbox}

Our focus will be on games that are \emph{generic}, by which we mean that for every $i$, and distinct action profiles $a$ and $a'$ that differ only in the $i$th index, either $a \succ_i a'$ or $a' \succ_i a$. Given integers $n\geq 2$ and $k \geq 2$, we use $\cG(n,k)$ to denote the set of all generic games with player set $[n]$ in which, for every $i \in [n]$, player $i$ has action set $[k]$. Since we are working with ordinal games, for a fixed $n$ and $k$, the set $\cG(n,k)$ is finite. 

\subsection{Results on game connectivity}\label{sec:connectedness}
Our results on the connectedness of games are stated below. 
Our first result, \Cref{prop:large_k}, considers the case where the number of actions per player tends to infinity, while the number of players remains fixed.
We show that for fixed $n\geq 3$, among the games in $\cG(n,k)$ that have a pure Nash equilibrium, the fraction that are connected goes to $1-\zeta_n$ where $\zeta_n >0$ is a positive constant. Observe, however, that $\zeta_n$ is small and decays very quickly with $n$; for example, $\zeta_3 \approx 0.0132$ and $\zeta_4 \approx 0.00002$. Of course, if we were to let $n$ grow, then the limiting fraction of connected games would go to 1.

\begin{theorem}\label{prop:large_k}
For fixed $n \geq 3$,
\[
    \lim_{k \to \infty} \frac{|\{g \in \cG(n,k) : \text{ $g$ \textnormal{is connected} }\}|}{|\{g \in \cG(n,k) : \text{ $g$ \textnormal{has a pure Nash equilibrium} }\}|} = 1-\zeta_n,
    \]
where 
\[
\zeta_n = 1 - \frac{e^{-\lambda_n}(1 - e^{\lambda_n - 1})}{1 - e^{-1}}
\]
and $\lambda_n$ is the smallest positive solution $x$ to
\[
x^{\frac{1}{n}} = e^{(n-1) \left(x^{\frac{1}{n}} - 1\right)}.
\]
\end{theorem}

For the case of $n=2$, which is not considered in \Cref{prop:large_k}, the fraction above tends to 0 as $k \to \infty$. In other words, among the games in $\cG(2,k)$ that have a pure Nash equilibrium, the fraction that are connected goes to $0$ as $k$ grows. This follows directly from existing results which show that among the games in $\cG(2,k)$ that have a pure Nash equilibrium, the fraction that are weakly acyclic goes to $0$ as $k$ gets large (\citealp*{amiet2021better,heinrich2023best}). Note, however, that the connectivity of best-response graphs is different from that of better-response graphs when $n=2$: \citet*{amiet2021better} show that among generic 2-player games that have a pure Nash equilibrium, the fraction that are also better-response weakly acyclic tends to one as $k$ gets large. 

One corollary of \Cref{prop:large_k} is that if $n \to \infty$ sufficiently slowly as $k \to \infty$, then the fraction of generic games with a pure Nash equilibrium which are connected tends to 1.
This is closely related to the main result of \citet*[Theorem 1]{johnston2023game}, which shows the same phenomenon but in the case where $k$ stays small compared to $n$. That is, the fraction of generic games with a pure Nash equilibrium which are connected tends to 1 as $n \to \infty$ in both the case where $k$ is small compared to $n$ and the case where $k$ is much larger than $n$.
This naturally raises the question of whether this phenomenon truly depends on $k$ or is solely a consequence of the fact that $n \to \infty$.
Our second theorem confirms the latter: as long as $n \to \infty$, the proportion of games which have a pure Nash equilibrium but are not connected, vanishes.

\begin{theorem}\label{thm:inf_k} 
    \[
    \lim_{n \to \infty} \left(\inf_{k \geq 2} \frac{|\{g \in \cG(n,k) : \text{ $g$ \textnormal{is connected} }\}|}{|\{g \in \cG(n,k) : \text{ $g$ \textnormal{has a pure Nash equilibrium} }\}|} \right) = 1 .
    \]
\end{theorem}
We note that our proof of this theorem depends intrinsically on the result of \citeauthor*{johnston2023game}, and this result should be viewed as the combination of the work in this paper on many-action games and the work in \citet*{johnston2023game} on many-player games.
Moreover, although this theorem imposes no restriction on how fast $n$ must grow relative to $k$, this comes at the cost of losing the exponential convergence in $n$, and neither result is strictly stronger than the other.
We also note that, although the results are similar, the proofs in the cases of many-actions and many-players require different techniques.
The key to the proof of \citet*{johnston2023game} was the notion of a good vertex which has both high in-degree and high out-degree, and showing that these form a strongly connected component. In the case where $k$ is much larger than $n$, however, almost all the vertices are sources, so a different approach is needed.
Instead, we rely heavily on finding good cycles in the subgames where only the first two players update their action, and show that these form a strongly connected component.
Clearly, this does not apply to the regime where $k$ is fixed and neither approach works in both regimes.
\newpage
\Cref{prop:large_k} and \Cref{thm:inf_k} together imply that 
\[
\frac{|\{g \in \cG(n,k) : \text{ $g$ \textnormal{is connected} }\}|}{|\{g \in \cG(n,k) : \text{ $g$ \textnormal{has a pure Nash equilibrium} }\}|} \to 1
\]
in both iterated limits $\lim_{n \to \infty} \lim_{k \to \infty}$ and $\lim_{k \to \infty} \lim_{n \to \infty}$, and in the double limit $\lim_{(n,k) \to (\infty,\infty)}$.

To complete the picture, we note that the fraction of acyclic games among those that are generic games and have a pure Nash equilibrium goes to zero as $n$ and/or $k$ get large (see \citealp*{heinrich2023best}). There is therefore a split in large game properties, with acyclic and better-response acyclic games being rare, and connected, weakly acyclic, and their `better-response' counterparts being very common.

\subsection{Implications for adaptive dynamics in games}\label{sec:dynamics}

We now consider games played over time, and discuss some implications of our connectivity results for adaptive dynamics in games. 

We recall standard definitions regarding dynamics in games. A player $i$'s \emph{observation set} at time $t$, denoted $o_i^t$, consists of all the information that~$i$ can observe at time $t$. For each integer $k\geq 2$, let $O_k$ denote the set of all possible observation sets for a player with action set $[k]$. A \emph{strategy} for a player with action set $[k]$ is a function $f\colon O_k \rightarrow \Delta([k])$, where $\Delta([k])$ is the probability simplex over $[k]$. A \emph{dynamic} on $\cG(n,k)$ consists of a specification of what information enters into each player's observation set at each time, and a strategy $f_i$ with action set $[k]$ for each player $i$. The play of a game $g\in \cG(n,k)$ under a given dynamic begins at time $t=0$ at an initial action profile $a^0$ chosen arbitrarily. This informs each player's observation set $o_i^1$ according to the dynamic. At time $t=1$, each player updates their action (randomly) according to $f_i(o_i^1)$, and we denote the new (random) action profile by $a^1$. The play continues in this manner, with each player updating their action at $t=2$ according to $f_i(o_i^2)$ to produce an action profile $a^2$, and so on.

Adopting the terminology of \citet*{hart2013simple}, a \emph{simple adaptive dynamic} is a dynamic in which player $i$'s observation set contains at most their own preference relation $\succsim_i$ and the last period's action profile $a^{t-1}$.\footnote{Expressed in terms that are standard in the literature, simple adaptive dynamics are those that are `uncoupled', `1-recall', and `stationary'. A dynamic is \emph{uncoupled} if a player's observation set consists at most of their own preference relation and of the past history of play, it is 1-\emph{recall} if the past history of play is restricted to only the last period, and it is \emph{stationary} if their strategy is time-independent.}

We say that a dynamic on $\cG(n,k)$ \emph{converges almost surely to a pure Nash equilibrium} of a game $g\in \cG(n,k)$ if when $g$ is played according to the dynamic from any initial action profile, there almost surely exists $T < \infty$ and a pure Nash equilibrium $a^*$ of $g$ such that $a^t=a^*$ for all $t \geq T$.

The influential impossibility result of \cite*{hart2003uncoupled,hart2006stochastic} states that there is no simple adaptive dynamic for which play converges almost surely to a pure Nash equilibrium in every (generic) game that has one. By contrast, our connectivity results imply the following.

\begin{proposition}\label{prop:dynamics}
There is a simple adaptive dynamic for which play converges almost surely to a pure Nash equilibrium in all but a vanishingly small fraction of games in $\cG(n,k)$ that have one. The statement holds for both iterated limits and in the double limit in $n$ and $k$.
\end{proposition}
\noindent The proof is straightforward. The best-response dynamic with inertia,\footnote{In this dynamic, at each time $t$, each player is independently selected with some positive probability. A selected player chooses an action that is a best-response to the actions chosen at $t-1$ by all other players. A player who is not selected keeps playing the same action they played at $t-1$.} which is a simple adaptive dynamic, converges almost surely to a pure Nash equilibrium in every weakly acyclic game \citep*{young2004strategic}, and therefore in every connected game that has a pure Nash equilibrium. \Cref{prop:dynamics} immediately follows from combining this fact with \Cref{prop:large_k} and \Cref{thm:inf_k}. 

What happens to \Cref{prop:dynamics} if $n$ is fixed at some value $n \geq 3$? Then, by a similar argument, we obtain that, as $k \to \infty$, there is a simple adaptive dynamic for which play converges almost surely to a pure Nash equilibrium in all but a $\zeta_n$ proportion of games in $\cG(n,k)$ that have one. Since $\zeta_n$ is almost imperceptibly small, this result is consistent with \citet*{heinrich2023best} who observed, via simulation only, that asynchronous best-response dynamics in which the next player was chosen at random from among all players seemed to converge to a pure Nash equilibrium in `essentially' all generic games that have one when $k$ was large and $n$ was fixed at $n\geq 3$.

As observed in \cite*{johnston2023game}, several dynamics that are not `simple adaptive' dynamics are known to lead to equilibrium in weakly acyclic games \citep*{young1993evolution,marden2007regret,marden2009payoff}.  Similarly to the argument given above for \Cref{prop:dynamics}, \Cref{prop:large_k} and \Cref{thm:inf_k} imply that such dynamics lead to equilibrium in all but a vanishingly small fraction of generic games that have a pure Nash equilibrium.

\subsection{Related work in game theory}\label{sec:related}

Our paper is related to two different strands of the game theory literature: random games, and dynamics in games.

As we will see in the rest of the paper, we establish our results on the prevalence of connectivity properties in games by drawing games at random and determining the probability that the best-response graphs of such randomly drawn games have specific connectivity properties. The idea of drawing games at random dates back at least to \citet*{goldberg1968probability} and \citet*{dresher1970probability} who derived the distribution of pure Nash equilibria. The results were further developed in \citet*{powers1990limiting,stanford1995note,rinott2000number}. Nash equilibria in random games have also been studied in \citet*{barany2007nash,hart2008evolutionary,daskalakis2011connectivity,pei2023nash}, and dominance-solvability was studied in \citet*{alon2021dominance}. A growing literature studies dynamics in random games, e.g.~see \citet*{amiet2021pure,heinrich2023best,johnston2023game,collevecchio2024finding}. In particular, \citet*{mimun2024best,collevecchio2024basins,collevecchio2025findingnashequilibriumrandom,ashkenazigolan2025simultaneousbestresponsedynamicsrandom} have studied best-response dynamics in 2-player games with a large number of actions per player although, in these papers, the player's payoffs are drawn in a manner that either introduces correlation or ties, whereas the case we study here corresponds to $n$-player $k$-action random games in which payoffs are drawn in an i.i.d.\ manner from an atomless distribution.

Our paper is related to the literature on simple adaptive dynamics in games \citep*{fudenberg1998theory} and to impossibility results relating to the equilibrium convergence of such dynamics \citep*{young2009learning}. Simple adaptive dynamics include the well-studied classes of better and best-response dynamics \citep*{young2004strategic,sandholm2010population} but, as noted in \citet*{johnston2023game}, results on game connectivity have implications for the convergence of many different types of dynamics, including regret- and payoff-based dynamics \citep*{marden2007regret,marden2009payoff}, that are not necessarily simple adaptive dynamics.\footnote{\citet*{bichler2023computing,bichler2023learning} employ learning dynamics such as simultaneous online dual averaging, or SODA, to successfully compute Nash equilibria in auction settings in which no analytic solutions for the equilibria are known. \citet*{bichler2023computing} remark that ``the wide range of environments in which SODA converges [to a Nash equilibrium] is remarkable''; moreover, ``the fact that we do find equilibrium consistently in a wide variety of auction games demands a closer look'' \citep*{bichler2023learning}. The success of the dynamics employed in these papers is somewhat surprising because there are no theoretical guarantees that such dynamics must converge to an equilibrium. We consider the class of all generic games rather than auction games specifically, and SODA is not a simple adaptive dynamic (because it has $\infty$-recall rather than 1-recall). Nevertheless, our results can be seen as providing some reasons to expect dynamics to `typically' find pure Nash equilibria in (sufficiently large) generic games that have one because such games are typically connected.} As previously mentioned, \citet*{hart2003uncoupled,hart2006stochastic} showed that there is no simple adaptive dynamic that is guaranteed to converge to a pure Nash equilibrium in every game that has one.\footnote{Our focus throughout is on pure Nash equilibria. For the relationship between adaptive dynamics and other equilibrium notions see, for example, \citet*{hart2000simple} for correlated equilibria, and \citet*{vlatakis2020no} for mixed Nash equilibria.} This influential negative result has been elaborated on in \citet*{babichenko2012completely,jaggard2014self}; and more recently in \citet*{milionis2023impossibility}.

\section{Random subgraphs of directed Hamming graphs}\label{sec:random}

This section presents the main theorems of our paper on the connectivity properties of random subgraphs of directed Hamming graphs, and from which the results in the previous section follow. Sections~\ref{sec:overview} onwards present the proofs of our main theorems.

For non-negative integers $n$ and $k$, the \emph{Hamming graph $H(n,k)$} has vertex set $[k]^n$ and an edge between each pair of vertices which differ in exactly one coordinate. The \emph{directed Hamming graph} $\vv{H}(n,k)$ is the directed graph obtained from $H(n,k)$ by replacing each edge $\{u,v\}$ with a pair of anti-parallel directed edges $(u,v)$ and $(v,u)$. For $i\in[n]$, a \emph{line in dimension $i$} is a subset of $[k]^n$ of size $k$ in which all elements pairwise differ in only the $i$th coordinate, and a \emph{line} is simply a subset of $[k]^n$ which is a line in dimension $i$ for some $i\in[n]$.

We consider the random subdigraph $\Lnk$ of $\vv{H}(n,k)$ defined by the following process. For each line, independently pick one vertex in that line uniformly at random to be the \emph{winner} of that line. Form $\Lnk$ from $\vv{H}(n,k)$ by, for each line $L$, deleting the directed edges induced by $L$ whose endpoint is not the winner. For $x,y\in [k]^n$ we say that $x$ can \emph{reach} $y$ in $\Lnk$ if there is a directed path from $x$ to $y$ in $\Lnk$. Naturally, in this case we say that $y$ can be \emph{reached from} $x$. Note that every vertex can reach and be reached from itself. For a vertex $x\in[k]^n$ (or more generally a set of vertices $A\subseteq[k]^n$) we will say that a line $L$ of $[k]^n$ can reach $x$ (or $A$) in $\Lnk$ if the winner of $L$ can reach $x$ (or $A$) in $\Lnk$.

Our main result on $\Lnk$ is the following.
\begin{theorem}
\label{thm:n-small}
    Let $n = n(k) \geq 3$ be such that $n \leq k^{1/2 - \eps}$ for some $\eps > 0$, and let $p = (\eta_{n-1})^n$, where $\eta_x$ is the extinction probability of a Galton--Watson branching process with offspring distribution $\Poi(x)$. Denote the number of sinks in $\Lnk$ which can and cannot be reached by every non-sink by $X$ and $Y$ respectively.
    Then $X$ and $Y$ are asymptotically independent Poisson random variables with means $1-p$ and $p$. 
    That is,
        \[\lim_{k \to \infty} \left[\Prob \left(X = a, Y = b \right) - \frac{e^{-1} (1-p)^a p^b}{a! b!} \right]  = 0\]
    for all $a, b \in \mathbb{Z}^+$.
\end{theorem}

It is easy to deduce \Cref{prop:large_k} from this result. 
Indeed, the game corresponding to $\Lnk$ is connected if $X \geq 1$ and $Y = 0$, and since these are asymptotically independent Poisson random variables with means $1-p$ and $p$ respectively, the probability that the game is connected tends to $(1 - e^{p-1})e^{-p}$. The total number of sinks tends in distribution to a $\Poi(1)$ random variable so the probability that the game has at least one sink tends to $1 - e^{-1}$. 
Finally, by \Cref{lem:GW}, the extinction probability $\eta_{n-1}$ is the smallest non-negative solution to $x = e^{(n-1)(x-1)}$ and $p = (\eta_{n-1})^n$.

We remark that when $n \to \infty$ as $k \to \infty$ sufficiently slowly, the above theorem shows that every non-sink can reach every sink with high probability. In fact, the important factor is that $n \to \infty$, as shown by the following theorem.
This is equivalent to \Cref{thm:inf_k}.

\begin{theorem}
\label{thm:n-large}
    Let $A(k,n)$ be the event that every non-sink can reach every sink in $\Lnk$. Then, as $n \to \infty$,
\[\inf_{k \geq 2} \Prob(A(k,n)) \to 1.\]
\end{theorem}

\section{Proofs of Theorems \ref{thm:n-small} and \ref{thm:n-large}}\label{sec:overview}

In this section, we use a series of results (which we prove later) to give proofs of our main results (\Cref{thm:n-small,thm:n-large}).
We will also briefly sketch how we prove these technical results as we use them.

We will work with what we call ``good sinks''. These are sinks which can be reached from lots of different lines; the idea being that a sink should either be reachable from every non-sink or from only a small number of lines. In fact, ``lots'' can be taken to be fairly small, which allows us to preserve much of the randomness in $\Lnk$, even while revealing that a particular vertex is a good or bad sink. We make the following definition.

\begin{definition}
    A sink in $\Lnk$ is said to be \emph{$\eps$-good} if it can be reached from at least $k^{\eps}$ lines in $\Lnk$, and \emph{$\eps$-bad} otherwise. For clarity, we will often suppress the $\eps$ and refer to sinks as simply \emph{good} or \emph{bad}.
\end{definition}

We start by considering the number of good and bad sinks. To decide if a sink is good or bad, we explore backwards from the sink until we have either found every point which can reach the sink or at least $k^{\eps}$ lines which can reach the sink. Since there are not many sinks (the total number of sinks follows a $\Poi(1)$ distribution asymptotically) and $k^{\eps}$ is small, whether different sinks are good or bad should be ``roughly independent''. This means that we expect the number of good sinks to follow a $\Poi(1-p)$ distribution, while the number of bad sinks should follow a $\Poi(p)$ distribution, where $p$ is the probability that a given sink is bad.

To determine the value of $p$, we consider exploring backwards from a single sink. The number of lines at each step of the exploration process grows like a Galton--Watson branching process which starts with $n$ individuals (since the vertex is a sink) and has offspring distribution $\Bin((n-1)(k-1), 1/k)$. Since this binomial distribution is close to a Poisson distribution with mean $n - 1$ (when $k$ is large), we might hope that the branching process behaves much like one with offspring distribution $\Poi(n-1)$, at least when $n$ is small compared to $k$. If this branching process goes extinct, it is likely to happen when the tree is fairly small, and so reaching a total of $k^{\eps}$ individuals should be quite close to failing to go extinct. This suggests that the probability a sink is bad should be around $(\eta_{n-1})^n$, where we recall that $\eta_x$ is the extinction probability of a Galton--Watson branching process which starts with a single individual and has offspring distribution $\Poi(x)$. 

Formalising the above gives the following theorem, which we prove in Section~\ref{sec:number_of_good_and_bad_sinks}.

\begin{restatable}{theorem}{poisson}\label{thm:poisson}
 Fix $\eps < 1/2$. Let $n = n(k)$ be such that $n \leq k^\eps/\log(k)$, and let $p = p(k) = (\eta_{n-1})^n$. Let $X$ and $Y$ be the number of $\eps$-good and $\eps$-bad sinks in $\Lnk$ respectively. Then $(X, Y) \overset{d}{\to} (X', Y')$ as $k \to \infty$, where $X'$ and $Y'$ are independent Poisson random variables with rates $1- p$ and $p$ respectively.
\end{restatable}

We note that $p \to 0$ as $n \to \infty$, and moreover, if $n \geq k^{\eps}$, then every sink is an $\eps$-good sink. From these two observations, it is not hard to show that every sink is a $\delta$-good sink (for some constant $\delta > 0$) with high probability as $n \to \infty$. More explicitly, we will deduce the following basic quantitative result from the proof of \Cref{thm:poisson}.

\begin{restatable}{lemma}{noBadSinks}\label{lem:no_bad_sinks}
    For all $n, k \geq 2$, the probability that there is an $\eps$-bad sink is at most
    \[(\eta_{n-1})^n + \frac{11}{k^{1-2\eps}}.\]
\end{restatable}

We now turn our attention to showing that every non-sink can reach every good sink with high probability. To do this we will find a large strongly connected component in $\Lnk$ which is well spread throughout the graph. The key idea here is to split the graph into ``slices'' where only the first two coordinates vary, and to consider ``good cycles'' in the slices. We will show that these are ubiquitous and that they are all in the same connected component, giving the required strongly connected component. The cycles will need to be able to reach lots of vertices, which we guarantee by asking for them to be of size at least $\sqrt{k}$, and also to be easy to reach. For this, we ask that the cycle has a large ``basin'' of points that can reach the cycle.

Formally, we define slices and good cycles as follows.

\begin{definition}
    For $\mbf{a} = (a_1, \dots, a_{n-2}) \in [k]^{n-2}$, the \emph{slice} at $\mbf{a}$, denoted by $S(\mbf{a})$, is the subgraph of $\Lnk$ induced on $\{\mathbf{x} : x_{i+2} = a_i ~\forall i \in [n-2]\}$. 
\end{definition}

\begin{definition}
    A cycle contained wholly in a slice $S(\mbf{a})$ is a \emph{good} cycle if it is of length at least $\sqrt{k}$ and the number of vertices in $S(\mbf{a})$ which can reach the cycle (via directed paths in $S(\mbf{a})$) is at least $k^2/(800 \log (k))$. 
\end{definition}

To show that the good cycles are ubiquitous with high probability, we show that the probability that a single slice contains a good cycle is at least $1/100$. Since each slice contains a good cycle independently, there must be many good cycles and these will be fairly well spread throughout $[k]^n$.

\begin{restatable}{lemma}{goodCycle}\label{lem:good_cycle}
    The probability that a slice contains a good cycle is at least $1/100$ for all large enough $k$.
\end{restatable}

To show that they are all in the same component, we first show that, with high probability, any two good cycles which are in slices which differ in just one coordinate are connected. We then only need the slices which contain a good cycle to form a connected component of $H(n-2, k)$.

\begin{restatable}{lemma}{goodCyclesSameComponent}\label{lem:good-cycles-same-component}
    There is a constant $c > 0$ such that, for all large enough $k$ and $n \geq 3$, the probability that all good cycles are in the same strongly connected component is at least $1 - \exp( - c n \sqrt{k}/\log(k))$.
\end{restatable}

To finish the proof, we will need two more things to occur with high probability: we need every non-sink to reach a good cycle and every good sink to be reachable from some good cycle. The first of these is shown by exploring out from a non-sink until we have encountered lots of different slices. With some care, we can ensure that the exploration process has revealed nothing about many of the slices and we can then reveal the winners in these slices to see if the process has reached a good cycle. 

The second is shown in a similar manner, except by working backwards from the sink. Unfortunately, in this case, we do not manage to revealing information about the winners in many of the slices. However, in most slices we will have only revealed that a small subset of the vertices are sources, which is not an unlikely event, and this revelation has only a small effect on the distribution of the winner of each line within the slice. 

\begin{restatable}{lemma}{reachingCycle}\label{lem:reaching_cycle} 
    There is a constant $c > 0$ such that, for all $k$ sufficiently large and $n \geq 3$, the probability that every non-sink can reach a good cycle is at least $1 - k^{-c n}$.
\end{restatable}

\begin{restatable}{lemma}{cycleToSink}\label{lem:cycle_to_sink}
     There is a constant $c > 0$ such that, for all $k$ sufficiently large, $n \geq 3$ and $\eps < 1/6$, the probability that every good sink can be reached from a good cycle is at least $1 - e^{-ck^\eps}$.
\end{restatable}

The proof of \Cref{thm:n-small} is now immediate. 
\begin{proof}[Proof of \Cref{thm:n-small}]
    Fix some $\eps \in (0, 1/6)$, say $1/10$.
    \Cref{lem:cycle_to_sink} shows that with high probability every $\eps$-good sink can be reached from at least one good cycle, and \Cref{lem:good-cycles-same-component} shows that all of the good cycles are in the same strongly connected component, again with a vanishingly small failure probability. Therefore every good sink can be reached from every vertex in a good cycle with high probability. \Cref{lem:reaching_cycle} shows that (with high probability) every non-sink can reach a good cycle, and can therefore reach every good sink.
    This means that (up to a vanishingly small failure probability) the good (resp. bad) sinks are exactly the sinks which can (resp. cannot) be reached by every non-sink.
    Finally, \Cref{thm:poisson} shows that the distribution of the number of good/bad sinks is as claimed.
\end{proof}

To prove \Cref{thm:n-large} we will need to combine the results here with an earlier theorem of the authors. The following is a slight simplification of Theorem 3 from \citet*{johnston2023game} to only consider the case where each player has the same number of actions.

\begin{theorem}[\protect{\citet*[Theorem 3]{johnston2023game}}]\label{thm:main_grids}
    For every $\eps>0$ there exist $c,\delta>0$ such that for all integers $n,k \geq 2$ for which $k \leq \delta \sqrt{n/ \log(n)}$, with failure probability at most ${k}^{-cn}$, every vertex of $\Lnk$ can either be reached from at most $N \coloneqq(1+\eps)k\log(k)$ vertices, or from every non-sink.
\end{theorem}

The proof of \Cref{thm:n-large} is now straightforward.

\begin{proof}[Proof of \Cref{thm:n-large}]
    Set $\eps = 1$ and take $0 < c, \delta \leq 1$ as in the statement of \Cref{thm:main_grids}. Suppose that $n$ is large enough that \Cref{lem:good-cycles-same-component,lem:reaching_cycle,lem:cycle_to_sink} all apply with $k = \ceil{\delta \sqrt{n/\log(n)}}$. Suppose also that $n$ is large enough that $n \geq 2 \sqrt{n\log(n)}$. By \Cref{thm:main_grids}, with probability at least $1 - 2^{-cn}$, every non-sink can reach every sink in $\Lnk$ when $k \leq \delta \sqrt{n/\log(n)}$.
    
    When $k \geq \delta \sqrt{n/\log(n)}$ we will consider $(1/10)$-good sinks. By \Cref{lem:no_bad_sinks}, the probability that there is a bad sink is at most $(\eta_{n-1})^n + 11k^{-0.8}$.  Since we have lower bounded $k$, this is $(\eta_{n-1})^n + O(n^{-0.4} \log(n))$ which tends to 0 as $n \to \infty$. 
    Together \Cref{lem:good-cycles-same-component,lem:reaching_cycle,lem:cycle_to_sink} show that the probability that every non-sink can reach every $(1/10)$-good sink tends to 1 as $n \to \infty$, completing the proof.
    \end{proof}

\section{Probabilistic preliminaries}

In this section we give a series of standard results which will be used throughout the paper.

Many of our proofs involve an exploration process to find the vertices and lines that can reach or be reached from a given vertex. These closely resemble Galton--Watson branching processes, and in particular, we will show that the ``extinction'' probabilities of our branching processes are close to the extinction probabilities of suitable Galton--Watson branching processes. 
The following well-known lemma (see, for example, \citet*{lyons2017probability}, Proposition 5.4) will therefore be useful.

\begin{lemma}\label{lem:GW}
    The extinction probability of a Galton--Watson process with offspring distribution $\mu$ is the smallest non-negative solution $z$ to the equation
    \[
    G(z)=z,
    \]
    where $G$ denotes the probability generating function of $\mu$.
\end{lemma}

Recall that $G(1)=1$ and if $\mu(2)>0$, then $G$ is strictly convex on $[0,1]$. It follows that for such $\mu$, denoting the extinction probability by $\eta$, we have $\eta<z$ for all $z$ satisfying $G(z)<z$. We will use this to prove the following lemma.

\begin{lemma}\label{lem:GW_bin}
    If $x \in (0, 1/4)$, then the extinction probability of a Galton--Watson process with offspring distribution $\Bin(n,1-x)$ is less than $2x^n$ for all $n\geq 2$.
\end{lemma}
\begin{proof}
    The probability generating function for a $\Bin(n,p)$ random variable is given by $G(z)=(1-p+pz)^n$. Thus, by Lemma~\ref{lem:GW} and the discussion above, it is sufficient to show that if $x \in (0, 1/4)$, then
    \[
    (x+2(1-x)x^n)^n < 2x^n
    \]
    for all $n\geq 2$.   
    Both sides of this inequality are polynomials in $x$ and one can check the claim for $n = 2,3$. 
    To show that the claim holds for $n \geq 4$, we note that for $x \leq 1/e$ and $n \geq 4$, we have that $nx^{n} \leq x^{n - \log(n)} \leq x^2$. Hence,
    \begin{align*}
        \big(x+2(1-x)x^n\big)^n &= x^n + \sum_{k=1}^n \binom{n}{k}(2(1-x)x^n)^k x^{n-k}\\
        &\leq  x^n + \sum_{k=1}^n (2(1-x) n x^{n})^k x^{n-k}\\
        &\leq x^n + \sum_{k=1}^n (2(1-x)x^{n - \log(n)})^k x^{n-k}\\
        &\leq x^n + x^n \sum_{k=1}^n (2(1-x)x)^k\\
        &\leq 2x^n
    \end{align*}
    for $x \leq 1/e$ and $n \geq 4$.
\end{proof}

As well as the extinction probability, we will be interested in the total population of a branching process and, in particular, showing that this is unlikely to be too large.

\begin{lemma}
    \label{lem:galton-watson-tail}
    Let $T$ be a Galton--Watson branching process with branching distribution $\Poi(\mu)$ and starting population $k$, and let the total population be $Z$. Suppose $\mu \in (0,1)$. Then there is a constant $c_\mu$ depending only on $\mu$ such that 
    \[ \Prob(Z \geq t) \leq c_\mu \mu^{-k}(\mu e^{1-\mu})^t \]
    for all $t \geq 0$.
\end{lemma}
\begin{proof}
    By the Otter--Dwass formula (\citet*[Exercise 5.35]{lyons2017probability}) we have
    \begin{align*}
        \Prob(Z = n) &= \frac{k}{n} \Prob(\Poi(\mu n) = n -k)\\
    &= \frac{k}{n} \cdot \frac{(\mu n)^{n -k} e^{-\mu n}}{(n-k)!}\\
    &\leq \frac{k}{n \mu^k} \frac{(\mu n)^n e^{-\mu n}}{n!}.
    \end{align*} 
    Using that $n! \geq \sqrt{2 \pi n} (n/e)^n$, we obtain
    \[\Prob(Z = n) \leq\frac{k}{\sqrt{2\pi} n^{3/2} \mu^k} \left( \mu e^{1-\mu} \right)^n.\]
    Hence, by summing over $n \geq t$ (and noting that the probability that $Z = \infty$ is 0), we find
    \begin{align*}
        \Prob(Z \geq t) &\leq \frac{1}{\sqrt{2 \pi} \mu^k}\sum_{n \geq t} \left( \mu e^{1-\mu} \right)^n\\
        &=  \frac{1}{\sqrt{2 \pi}\mu^k}\cdot \frac{e^\mu (\mu e^{1-\mu})^{\lceil t\rceil }}{e^\mu - e \mu }.\qedhere
    \end{align*}
\end{proof}

To prove \Cref{thm:poisson} we will need to show that the probability that a Galton--Watson branching process goes extinct with a large total population is small.
We use the following standard lemma (see \citet*{lyons2017probability}, Proposition 5.28(ii)).

\begin{lemma}\label{lem:extinction_branching_distribution}
    Let $T$ be a Galton--Watson branching process whose offspring distribution has probability generating function $G$. Then the law of $T$ conditioned on extinction is that of a Galton--Watson branching process whose offspring distribution has probability generating function $G(qs)/q$, where $q$ is the extinction probability of $T$.
\end{lemma}
\begin{corollary}
    Suppose $T$ is a Galton--Watson branching process with offspring distribution $\Poi(\lambda)$ for some $\lambda > 1$. Then the law of $T$ conditioned on extinction is that of a Galton--Watson branching process with offspring distribution $\Poi(\nu)$ where $\nu < 1$ satisfies
    \[\nu e^{-\nu} = \lambda e^{-\lambda}.\]
\end{corollary}

We will also need the following tail bounds on the hypergeometric random variable and on the sum of Bernoulli random variables. The first is an immediate corollary of the stronger bounds in \citet*{serfling1974probability}, and a good description of the second can be found in \citet*{mitzenmacher2017probability} (see Theorem 4.5).

\begin{lemma}[\protect{\citet*[Corollary 1.1]{serfling1974probability}}]
\label{lem:hyper}
    Let $X \sim \Hypergeometric(N, K, n)$. Then, for all $0 < t$, we have
    \[\mathbb{P}\left( X \leq (K/N  - t)n \right) \leq e^{-2t^2n}.\]
\end{lemma}

\begin{lemma}
\label{lem:chernoff}
    Let $X_1,\dots,X_n$ be independent Bernoulli random variables. Let $X$ denote their sum, and let $\mu=\expec[X]$. Then, for all $\delta>0$, we have
    \[
    \Prob(X\leq (1-\delta)\mu)\leq e^{-\delta^2\mu/2}.
    \]
\end{lemma}

In order to determine that the number of good sinks and bad sinks are asymptotically distributed like independent Poisson random variables, we will use the ``method of moments''.
That is, we will show that the factorial moments of the number of good and bad sinks tend to the anticipated values and appeal to the following theorem
(see, for example, \citet*{bollobas2001random}, Theorem 1.23).

\begin{theorem}
\label{thm:method_of_moments}
    Let $\lambda = \lambda(k)$ and $\mu = \mu(k)$ be non-negative bounded functions on $\mathbb{N}$. For each $k$, let $X(k)$ and $Y(k)$ be non-negative integer values random variables defined on the same space. Suppose for all $a, b \in \mathbb{Z}^+$ we have 
    \[\lim_{k \to \infty} \expec \left[X(X-1) \cdots (X-a+1) \cdot Y(Y-1)\cdots (Y - b + 1) - \lambda^a \mu^b\right]  = 0 .\]
    Then $X$, $Y$ are asymptotically independent Poisson random variables with means $\lambda$ and $\mu$. That is,
        \[\lim_{k \to \infty} \left[\Prob \left(X =a, Y = b \right) - \frac{e^{-(\lambda + \mu)} \lambda^a \mu^b}{a! b!} \right]  = 0\]
    for all $a, b \in \mathbb{Z}^+$.
\end{theorem}

In our case, the values $\lambda(k)$ and $\mu(k)$ are related to the extinction probabilities of Poisson Galton--Watson branching processes. As part of this, we will show that the extinction probability of a suitable Galton--Watson branching process with a binomial offspring distribution is close to that of a Galton--Watson branching process with a Poisson offspring distribution. For this, we will need the following two results on the total variation.

\begin{lemma}[Folklore]
For any two random variables $X$ and $Y$, there exists a coupling of $X' \sim X$ and $Y' \sim Y$ such that
\[ \Prob(X' \neq Y') = d_{TV}(X, Y).\]
\end{lemma}

\begin{theorem}[Le Cam's Theorem \citet*{le1960approximation,steele1994cam}]
    \label{thm:le-cam}
    Let $X_1, \dots, X_n$ be independent Bernoulli random variables with success probabilities $p_1, \dots, p_n$. Let $S  = X_1 + \dotsb + X_n$ and let $\mu$ denote the expectation of $S$ (i.e.
    \(\mu = \expec[S] = \sum_{i=1}^n p_i\)).
    Then
    \[ \sum_{k=0}^\infty \left| \mathbb{P}(S = k) - \frac{\mu^k e^{-\mu}}{k!} \right| < 2 \min \left\{1, \frac{1}{\mu}\right\} \sum_{i=1}^n p_i^2. \]
\end{theorem}

We end this section with the following folklore result (a short combinatorial proof can be found in~\citet*{mcdiarmid2021component}).

\begin{lemma}\label{lem:trees}
    If $G$ is a graph with maximum degree $\Delta$, then for each $m\in \N$ there are at most $(e\Delta)^{m-1}$ trees of order $m$ in $G$ that contain a given vertex.
\end{lemma}

\section{The number of good and bad sinks}
\label{sec:number_of_good_and_bad_sinks}

In this section we prove Theorem~\ref{thm:poisson}. For this, we will explore out from selected vertices which we assume are sinks and approximate the probability that the sinks are good/bad. We will then use the method of moments (Theorem \ref{thm:method_of_moments}) to complete the proof. When exploring from a single sink the process grows like a Galton--Watson branching process which starts with $n$ individuals (since the vertex is a sink) and has offspring distribution $\Bin((n-1)(k-1), 1/k)$. Since this binomial distribution is close to a Poisson distribution with mean $n - 1$, we might hope that the branching process behaves much like one with offspring distribution $\Poi(n-1)$, at least when $n$ is small compared to $k$. This is indeed the case (at least for our purposes) and we start this section by proving this.

\begin{lemma}\label{lem:pois_approx}
    Let $n  = n(k) \geq 3$ be such that $n \leq k^{\eps}/\log(k)$, where $\eps \leq 1/2$, and let $T$ be a Galton--Watson process starting with $n$ individuals and with branching distribution $\Bin((n-1)(k-1), 1/k)$.
    Let $\lambda(k)$ be the probability that the total population is at most $k^{\eps}$.
    Then, as $k \to \infty$,
    \[| \lambda(k) -  (\eta_{n-1})^n| \to 0.\]
\end{lemma}

\begin{proof}
    Let $\{(X_{i}, X_{i}') : i \geq 1 \}$ be pairs of random variables such that each pair is independent of the other pairs, $X_{i} \sim \Bin((n-1)(k-1), 1/k)$,  $X_{i}' \sim \Poi(n-1)$ and $X_i$ and $X_i'$ are coupled so that 
        \[\Prob(X_{i} \neq X_{i}') = d_{TV}(X_{i}, X_{i}').\]
    We may decompose $X_{i}'$ as the sum of a $\Poi((n-1)(k-1)/k)$ and a $\Poi((n-1)/k)$ random variable. By Le Cam's Theorem, the total variation distance between $X_{i}$ and the former Poisson distribution is at most $2(n-1)/k$. The overall total variation distance is therefore at most 
    \[2(n-1)/k + 1 - e^{-(n-1)/k} \leq 3(n-1)/k.\]
    Define $T = (T_t)_{t \geq 0}$ by $T_0 = n$ and \[T_{t+1} = \sum_{i=Z_{t-1} + 1}^{Z_{t-1} + T_t}X_{i}\] where $Z_{t-1} = \sum_{i=0}^{t-1} T_i$, so that $T$ is a Galton--Watson process starting with $n$ individuals and with branching distribution $\Bin((n-1)(k-1), 1/k)$.
    Similarly, define the process $T' = (T'_t)_{t \geq 0}$ by $T'_0 = n$ and \[T'_{t+1} = \sum_{i=Z'_{t-1} + 1}^{Z'_{t-1} + T'_t}X'_{i}\] where $Z'_{t-1} = \sum_{i=0}^{t-1} T'_i$, so that $T'$ is a Galton--Watson process starting with $n$ individuals and with branching distribution $\Poi(n-1)$.
    Let $\tau$ and $\tau'$ be the total population of the processes $T$ and $T'$ respectively, and note that we are interested in the event that $\tau < k^{\eps}$.
   
    We certainly have
        \[\left| \Prob(\tau < k^{\eps}) - \Prob(\tau' < k^{\eps}) \right| \leq \Prob(\exists i < k^{\eps} \text{ s.t. }X_i \neq X_{i}')\leq \frac{3(n-1)}{k} \cdot k^{\eps} \to 0.\]
    
    By \Cref{lem:extinction_branching_distribution}, the law of $T'$ conditioned on extinction is that of a branching process with branching distribution $\Poi(\nu)$ where $\nu < 1$ satisfies $\nu e^{-\nu} = (n-1) e^{-(n-1)}$. In particular, the conditioned process is dominated by a process with branching distribution $\Poi(1/2)$ (for $n \geq 3$). 
    By \Cref{lem:galton-watson-tail}, the probability that the total population of this conditioned process is at least $k^{\eps}$ is at most
    \[c_{1/2} 2^{k^{\eps}/\log(k)} (\sqrt{e}/2)^{k^\eps}\]
    and this clearly tends to 0 as $k \to \infty$.
    Hence, $\Prob(\tau' < k^{\eps} | \tau' < \infty) = (1 - o(1))$ and
    \[
        \lim_{k\to \infty}\Prob(\tau' < k^{\eps})  =  \lim_{k\to \infty}\Prob(\tau' < k^{\eps} | \tau' < \infty) \Prob(\tau' < \infty) =  (1 - o(1)) (\eta_{n-1})^n. \qedhere
    \]
\end{proof}

We now turn to the proof of Theorem \ref{thm:poisson}.

\poisson*

\begin{proof}
    First, note that the expectation
    \[\expec \left[X(X-1) \cdots (X-a+1) \cdot Y(Y-1)\cdots (Y - b + 1)\right]\]
    is the sum over all ordered choices of distinct points $x_1, \dots, x_a$ and distinct points $y_1, \dots, y_b$ (where we may have $x_i = y_j$) of the probability that $x_1, \dots, x_a$ are good sinks and $y_1, \dots, y_b$ are bad sinks. 

    If it is not the case that all of the points are distinct and no two share a line, then the probability that $x_1, \dots, x_a$ are good sinks and $y_1, \dots, y_b$ are bad sinks is 0. Indeed, either there is a point which needs to be both a good sink and a bad sink, which is impossible, or there are two points on the same line which both need to be sinks and both need to win the shared line. 
    If it is the case that all of the points are distinct and no two share a line, then each of them are sinks independently with probability $(1/k)^n$ and the probability that they are all sinks is $(1/k)^{n (a + b)}$. 

    Suppose that the points are all distinct and no two share a line, and condition on the fact that they are all sinks. We describe a backwards exploration process which checks if the given sinks are good, and an approximate coupled process which is easier to work with. The coupled process is distributed like $(a+b)$ independent Galton--Watson branching processes each of which start with $n$ individuals and have branching distribution $\Bin((n-1)(k-1), 1/k)$. Hence, by Lemma~\ref{lem:pois_approx}, we know that the probability the approximate coupled process finds that $x_1, \dots, x_a$ are good sinks and $y_1, \dots, y_b$ are bad sinks tends to $p^a (1-p)^b$ as $k \to \infty$, where $p  = (\eta_{n-1})^n$. We will show that the probability that the approximate process and the exact process disagree at any point is $o(1)$, which essentially completes the proof.

    The exact process explores from the points $x_1, \dots, x_a$ in turn and then from $y_1, \dots, y_b$ in turn. Suppose we are exploring from the point $z$. We start with the $n$ lines which contain $z$ and we mark the point $z$ on each line. We will call these lines \emph{active}, and we remark that we know the winner of each active line and that each active line can reach $z$. At each step, pick the first active line $\ell$ (according to some fixed ordering) and, for each point $w$ on $\ell$ which is not the winner of $\ell$, check if $w$ is the winner on each of the lines containing $w$ (except $\ell$). If $w$ is the winner of such a line, we add the line to our set of active lines, and we note that we know the winner of the line and that the line can reach $z$. Once we are done exploring $\ell$, we remove $\ell$ from the set of active lines. 
    If at any point in this process, we have found at least $k^{\eps}$ lines that can reach $z$, we know $z$ is a good sink and we continue to the next sink as appropriate. Note that we may stop part way through exploring a line $\ell$.
    If we run out of active lines before finding $k^{\eps}$ lines that can reach $z$, then $z$ must be a bad sink.

    We now describe the coupled process which behaves in a similar way, except that we may add phantom lines. The coupled process also explores from the points $x_1, \dots, x_a$ in turn and then from $y_1, \dots, y_b$ in turn. Suppose we are exploring from the point $z$, and start with the same set of active lines as before. At each step, pick the first active line $\ell$, and suppose first that $\ell$ is not a phantom line. For each point $w$ on $\ell$ which is not the winner of $\ell$, consider each of the lines containing $w$. If we are considering a line for which we already know the winner (and it is therefore not $w$), we add a phantom copy of the line to the set of active lines with probability $1/k$ and we mark $w$ as the winner of this phantom copy. 
    If we already know $m \geq 0$ points on $\ell$ are not the winner, then the probability that $w$ is the winner is $1/(k-m)$. We check if $w$ is the winner of the line. If it is not, we do nothing. If it is, we add the line to the set of active lines with probability $(k-m)/k$, so the overall probability of adding the line is $1/k$. 
    If the line $\ell$ is a phantom line, we again consider each point $w$ on $\ell$ which is not the winner and each of the lines (except $\ell$) which are incident to $w$. For each line we consider, we will add a phantom copy of that line with $w$ as the winner with probability $1/k$. After we have fully processed $\ell$, we remove it from the set of active lines.
    If at any point in this process, we have found at least $k^{\eps}$ lines that can reach $z$, we guess $z$ is a good sink and we continue to the next sink as appropriate.
    Alternatively, we may stop when there are no active lines.
    Each line we check is added with probability $1/k$ and we check $(n-1)(k-1)$ lines when checking $\ell$, and so the probability that this coupled process finds less than $k^{\eps}$ lines which can ``reach'' $z$ has probability $\lambda(k)$ (defined as in the statement of Lemma~\ref{lem:pois_approx}).
    
    We now show that with probability $(1- o(1))$ the coupled process and the exact process agree.
    There are two ways the processes could diverge, we could add a phantom line in the approximate process as we already know the winner of the line is somewhere else, or we could discover $w$ is the winner of the line and still not add it.
    There are at most $(a+b)\lceil k^{\eps} \rceil$ lines that are ever explored and any unexplored line is checked at most once per line that is explored, so the probability that we discover our test vertex $w$ is the winner of a line and reject the line is at most $(a+b)\lceil k^{\eps} \rceil/k$. The probability that we reject any of the first $(a+b)\lceil k^{\eps} \rceil$ lines we should add is at most $(1+o(1))(a+b)^2 k^{2 \eps - 1} = o(1)$. 

    Suppose that we never reject a line that we should add. Then we reveal the winners of at most $(a+b)\lceil k^{\eps} \rceil$ lines, and we can check each of these lines at most once per line we explore. We explore at most $(a+b)\lceil k^{\eps} \rceil$ lines, so we certainly check a line for which we already know the winner at most $((a+b)\lceil k^{\eps} \rceil)^2$ times. The probability that any of these are accepted is at most $(1+o(1))(a+b)^2k^{2 \eps - 1} = o(1)$. Hence, with probability $1 - o(1)$, the two processes agree and the probability that $x_1, \dots, x_a$ are good sinks and $y_1, \dots, y_b$ are bad sinks is $(1 - \lambda(k))^a \lambda(k)^b + o(1)$, conditional on them being sinks.

    It remains to put the pieces together and verify the condition of Theorem~\ref{thm:poisson}. For this, we note that there are
    \[k^n (k^n  - (n(k-1) + 1))(k^n  - 2(n(k-1) + 1)) \cdots (k^n  - (a+b - 1)(n(k-1) + 1)) = (1 + o(1)) k^{n (a+b)}\]
    ways of picking $a+b$ ordered points with no two sharing a line.
    \begin{align*}
         &\expec \left[X(X-1) \cdots (X-a+1) \cdot Y(Y-1)\cdots (Y - b + 1)\right]\\
         &\quad\quad = (1+o(1))k^{n(a+b)} \cdot (1/k)^{n(a+b)} \cdot ((1 - \lambda(k))^a \lambda(k)^b + o(1))\\
         &\quad\quad= (1 - p)^a p^b + o(1). \qedhere
    \end{align*}
    Hence, we can appeal to \Cref{thm:method_of_moments} to finish the proof.
\end{proof}

From the above proof we see that the probability that a vertex $x$ is a bad sink is at most $1/k^n \cdot (\lambda(k) + 8\ k^{2\eps - 1})$ for large $k$, where we have used that $\lceil k^{\eps} \rceil^2 / k \leq 4 k^{2\eps - 1}$ for $k \geq 1$. We also note that, from the proof of \Cref{lem:pois_approx}, we have
\[\lambda(k) = \Prob(\tau < k^{\eps}) \leq \Prob(\tau' < \infty ) + \frac{3(n-1)}{k^{1-\eps}} = (\eta_{n-1})^n + \frac{3(n-1)}{k^{1-\eps}}. \]
In particular, the expected number of bad sinks is at most
\[
    (\eta_{n-1})^n + \frac{3(n-1)}{k^{1-\eps}} + 8k^{2\eps - 1}.
\]
When $n \geq k^{\eps}$, there are no bad sinks, so in fact we may bound the expectation by $(\eta_{n-1})^n +11k^{2\eps - 1}$.
\Cref{lem:no_bad_sinks} now follows from an application of Markov's inequality.

\section{Good cycles are found in lots of slices}

The aim of this section is to prove Lemma~\ref{lem:good_cycle} which shows that good cycles are ubiquitous. 

\goodCycle*
Let us first give a rough sketch of the proof of this lemma.
We start by considering the exploration process which follows the out-edges from a point in the slice, starting with a horizontal edge. 
The coordinates of the active point alternate updating randomly until one of the coordinates repeats, at which point the process starts to cycle. 
Since the new points are drawn uniformly, the time until the first repeat (in a particular coordinate) is a `birthday problem' and we therefore expect the first repeat to be when we have $\Theta(\sqrt{k})$ values.
Given that the next update is a repetition, it is uniformly random which of the previous values it takes.
If the repeated value was first seen early in the sequence, the cycle must be quite long and we expect a constant proportion of the points to be in the basin of a long cycle.
However, we need many points which are in the basin of the \emph{same} long cycle.
To this end, we bound the number of long cycles in the slice to show that at least one of them must have a large basin.

We now make the above sketch precise.

\begin{proof}[Proof of \Cref{lem:good_cycle}]
    Fix a slice $S$ and let us suppress all coordinates except the first two, so we are working on a $k \times k$ grid with vertical and horizontal directions. 
    Let us call a cycle in $S$ a \emph{long cycle} if it has length at least $\sqrt{k}$.
    We start by showing that the expected number of vertices in $S$ which are in the basin of at least one long cycle is at least $k^2/9$. This implies that the probability that there are least $ \delta k^2$ such vertices is at least $(9\delta - 1)/(9 ( \delta - 1))$. 
    Indeed, if $X$ is the number of vertices in the basins of the long cycles, we have
    \[\frac{k^2}{9} \leq \expec[X] \leq \delta k^2 ( 1 - \Prob(X \geq \delta k^2)) + k^2 \Prob(X \geq \delta k^2),\]
    and this implies $\Prob(X \geq \delta k^2) \geq (9\delta - 1)/(9 ( \delta - 1))$. 
    We will later take $\delta =  1/20$.
    
    Fix a point $(x_0, y_0) \in [k]^2$ and consider the exploration process which begins by following the horizontal edge from $(x_0, y_0)$ to reach some point $(x_1, y_0)$, then the vertical edge from $(x_1, y_0)$ to reach some point $(x_1, y_1)$, and continues alternating between vertical and horizontal edges. At some point this exploration process must start cycling, and we will show that with probability at least $1/9$ this cycle has length at least $\sqrt{k}$. We let the sequence of $x$-coordinates be $x_0, x_1, x_2 \dots$, and similarly let $y_0, y_1, \dots$ be the sequence of $y$-coordinates. If $y_i \neq y_j$ for all $j < i$, then $x_{i+1}$ is uniformly distributed over $[k]$, while if $y_i = y_j$ for some $j < i$, then $x_{i+1} = x_{j+1}$ and the process is in a loop. The sequence of $y$-coordinates behaves similarly except the process does not necessarily loop when repeating $x_0$ as we have not revealed anything about the vertical line through $(x_0, y_0)$. 
    The probability that $x_0, \dots, x_{m-1}$ are all distinct and $y_0, \dots, y_{m-1}$ are all distinct is the square of the probability that $m$ uniform samples from $[k]$ are all distinct.
    The probability that the $i$th sample is the same as an earlier sample is $(i-1)/k$, so the probability that $m$ samples are all distinct is at least $1 - \frac{m(m-1)}{2k}$.
    Now condition on the event that $x_1, \dots, x_{i-1}$ are all distinct, $y_0, \dots, y_{i-1}$ are all distinct but $x_i \in \{x_1, \dots, x_{i-1}\}$. In this case, $x_i = x_j$ for some $j$ which is uniformly distributed over $\{1, \dots, i -1\}$. The length of the cycle is then 1 if $j = i -1$ and $2(i-j)$ otherwise. The same is true if instead $y_i \in \{y_0, \dots, y_{i-1}\}$, and so the probability that the cycle is of length at least $i$ given that the first repeat is either $x_i$ or $y_i$ is at least $1/2 - 1/i$ when $i \geq 3$.
    Hence, the probability that the exploration process reaches a cycle of length at least $\sqrt{k}$ is at least 
    \[\left( \frac{1}{2} - \frac{1}{\lceil\sqrt{k} \rceil} \right) \cdot \left(1 - \frac{\lceil \sqrt{k}\rceil(\lceil \sqrt{k}\rceil-1)}{2k}\right)^2\]
    which is at least $1/9$ for large enough $k$.

    We now show that there cannot be many cycles of length at least 4. For $r \geq 2$, the expected number of cycles of length $2r$ is \[\frac{(k(k-1)\dotsb (k - r + 1))^2}{rk^{2r}} \leq 1/r.\]
    Hence, the expected number of cycles is at most $\sum_{r=2}^{k^2} 1/r \leq 2 \log (k)$. By Markov's inequality, the probability that there are at least $\lambda \log (k)$ cycles is at most $2/\lambda$. 
    With probability at least $(9\delta - 1)/(9 ( \delta - 1)) - 2/\lambda$, at least $\delta k^2$ of the vertices are in the basins of long cycles and there are at most $\lambda \log (k)$ cycles, which together imply that there is some long cycle with a basin of size at least $\frac{\delta k^2 }{\lambda\log(k)}$.
    To get the claimed result, we set $\delta = 1/20$ and $\lambda  = 34200/929$ to find that with probability at least $1/100$ there is a long cycle with a basin of size at least $929 / 684000 \geq 1/800$.
\end{proof}

\section{All good cycles are in the same strongly connected component}

The main aim of this section is to prove Lemma~\ref{lem:good-cycles-same-component}, that is, to prove that all good cycles are in the same strongly connected component.
We will first show that every pair of good cycles in slices which differ in exactly one coordinate are in the same strongly connected component with high probability. 
This does not immediately imply that all good cycles are in the same strongly connected component.
For example, it is possible that there is a slice containing two good cycles while none of the slices differing in a single coordinate contain a good cycle. 
To ensure that all good cycles are in the same strongly connected component, we want the slices which contain a good cycle to `form a connected component in $H(n-2, k)$ with at least two vertices'.
Fortunately, the fact that each slice contains a good cycle independently with probability at least $1/100$ is more than sufficient for this to happen with high probability.

Our first step towards showing that every pair of good cycles in slices which differ in exactly one coordinate can reach each other is to consider the case where $n = 3$.
This case actually does the majority of the heavy lifting and extending the result to all $n \geq 4$ only requires a little extra work. However, as part of this, we will need to apply the $ n = 3$ case after having revealed the winners in many of the slices (which we do to check if they contain a good cycle).
For this reason, we avoid $z$-coordinates which are even (and at least three) in the lemma below, leaving us free to check these slices for a good cycle.

\begin{lemma}\label{lem:cycle_to_cycle_3}
    Suppose that $n = 3$ and that $S(1)$ and $S(2)$ contain good cycles $C_1$ and $C_2$ respectively. Then, for any choice of winners within $S(1)$ and $S(2)$, the probability that there is a path from $C_1$ to $C_2$ in $\Lnk$ is
    \[1 - \exp(-\Omega(\sqrt{k}/\log(k))).\]
    Furthermore, this path does not include any vertex $(x,y,z)$ for which $z \geq 3$ is even.
\end{lemma}

To prove this lemma we start by exploring in the $z$-direction from the $\Theta(\sqrt{k})$ vertices in $C_1$ to get points in $\Theta(\sqrt{k})$ different slices. 
We then explore within each of these slices to get $\Theta(k)$ points which can all be reached from $C_1$.
This is then repeated once more to again boost the number of points which can be reached from $C_1$. That is, we explore in the $z$-direction to get to $\Theta(k)$ slices and we explore in each of these slices to get $\Theta(k^{3/2})$ points which can all be reached from $C_1$.
Finally, we explore in the $z$-direction from each of these points.
Heuristically, this final exploration should give $\Theta(\sqrt{k})$ points in $S(2)$ and each of these should be in the basin of $C_2$ with probability $\Theta(1/\log(k))$. This means that we expect $\Theta(\sqrt{k}/\log(k))$ of the points to be in the basin of $C_2$, and it is highly likely we have reached at least one point in the basin of $C_2$.
We now make this argument rigorous.

\begin{proof}[Proof of \Cref{lem:cycle_to_cycle_3}]   
    Let $m = \lceil \sqrt{k} \rceil$ and pick $m$ vertices from $C_1$ to form $A_0$.
    We will label the three directions by $x$, $y$ and $z$, where a slice is spanned by the $x$- and $y$-directions.
    We start by exploring in the $z$-direction from each of the $m$ points in turn, and we form $A_1$ by keeping the new points which are the first in a slice $S(z)$ where $z \geq 3$ is odd.
    The number of points we keep dominates a $\Bin(m, (k-2m - 2)/(2k))$ random variable. Hence, using Lemma~\ref{lem:chernoff}, the probability that there are at most $m/3$ points is at most
    \( \exp ( - m/360) \)
    for $k$ large. 
    If there are less than $m/3$ points in $A_1$, we will declare a failure, and so we assume that $|A_1| \geq \lceil m/3 \rceil$.

    From each of the points in $A_1$, we explore in the slice by first moving in the $x$-direction, then alternating between moving in the $y$-direction and the $x$-direction until an $x$- or $y$-coordinate is repeated. For each point $\mathbf{v} = (x_0, y_0, z)$ in $A_1$, we can sample this as follows. Generate a sequence of random points $(x_1, y_1), (x_2, y_2), \dots$ by picking points independently and uniformly at random from all $k^2$ options. The exploration process then follows the sequence 
    \[  (x_0, y_0, z),
        (x_1, y_0, z),
        (x_1, y_1, z),
        (x_2, y_1, z), \dots \]
    until an $x$- or $y$-coordinate is repeated for the first time, at which point the exploration process stops.
    If there is a repeated $x$- or $y$-coordinate in the points
    \[(x_0, y_0, z), (x_1, y_1, z), \dots , (x_m, y_m, z),\]
    we move onto the next point from $A_1$. 
    Otherwise, add the points $(x_i, y_i, z)$ where $1 \leq i \leq m$ to the set $A_2$.
    The probability that we have added points to $A_2$ when processing $v$ is at least 
    \[\left(1- \frac{m(m+1)}{2k}\right)^2 \geq \frac{1}{5}\]
    for $k$ large enough. 
    As each slice is independent, the number of slices that contribute to $A_2$ is dominated by a $\Bin(\lceil m/3\rceil, 1/5)$ random variable and, by Lemma~\ref{lem:chernoff}, the probability that there are less than $m^2/20$ points in the set $A_2$ is at most $\exp(- m/480)$.

    This has given us lots of point in $A_2$, but the next step is to explore again in the $z$-direction and we actually need to have lots of different pairs of $x$- and $y$-coordinates.
    Recall that all of the points added are of the form $(x_i, y_i, z)$ where $(x_i, y_i)$ was picked independently and uniformly at random. However, we know that the points have been added to $A_2$ and, conditional on this information, the points are not independent.
    However, collisions should still be very rare and it is enough for us to bound the total number of collisions across the first $m$ pairs $(x_i, y_i)$ picked for all of the at most $m$ points in $A_1$.   
    The probability that a picked point matches an earlier point or one of the points from $A_0$ is at most $(m^2 + m)/k^2$, and so the number of such points is dominated by a $\Bin(m^2, (m^2 + m)/k^2)$ random variable. The probability that this is at least $\ell$ is at most 
    \[\binom{m^2}{\ell} \left(\frac{m^2 + m}{k^2}\right)^\ell \leq \left(\frac{e m^3 (m + 1)}{\ell \cdot k^2} \right)^\ell. \]
    Substituting in $\ell = \lceil m^2/100\rceil$, we find that the probability there at least $m^2/100$ collisions is $O(k^{-\Omega(k)})$.
    We now form $A_3$ by going through the points in $A_2$ and adding the point to $A_3$ if it does not have the same $x$- and $y$-coordinates as a point already in $A_3$ or in $A_0$.
    We also stop the process as soon as $A_3$ contains $\lceil m^2/25\rceil$ points. 
    Combining the above estimates, the probability that $A_3$ has $\lceil m^2/25\rceil$ points is at least $1 - e^{-\Omega(m)}$, and we assume that this is the case. 

    We now repeat the process with the points from $A_3$ replacing the points in $A_0$. That is, we explore in the $z$-direction from each of the points in $A_3$ and form the set $A_4$ by keeping the points $(x,y,z)$ which are in a slice we have not visited before and for which $z$ is odd and at least three.
    At any point in this exploration the number of slices which we will not accept is at most 
    \[\frac{k}{2} + m + 1 + \left \lceil \frac{m^2}{25} \right \rceil = \frac{27}{50} k + O(\sqrt{k}).\]
    Hence, the number of points in $A_4$ is dominated by a binomial random variable and we can use Lemma~\ref{lem:chernoff} to find that the probability that $A_4$ has at most $k/60$ points is at most $\exp(-\Omega(k))$.

    From each point $v$ in $A_4$, we explore out from $v$ in its slice by first moving in the $x$-direction, then alternating between moving in the $y$-direction and the $x$-direction until we repeat a coordinate. We do this by sampling $m$ points from $[k]^2$ uniformly at random as before.
    The probability that no two of the $m$ points share a coordinate with each other or with $v$ is at least $1/5$. We now condition on this being the case and go through the $m$ points in turn. We say a point $(x,y,z)$ is \emph{good} if 
    \begin{enumerate}
        \item the point $(x,y,2)$ is in the basin of $C_2$,
        \item we have not already encountered a point of the form $(x, y, \star)$, 
        \item the winner of the line through the point in the $z$-direction is $(x,y,2)$.
    \end{enumerate}
    If any of the $m$ points are good, then we have the requisite path from $C_1$ to $C_2$.
    The distribution of the $i$th point is uniform over all points which do not share an $x-$ or $y-$coordinate with the first $i-1$ points or with $v$. Each banned value for a coordinate reduces the number of points which we may take for $i$ by at $k$, and so the conditioning reduces the number of choices for the $i$th point by at most $2km$.
    When exploring within a slice we encounter at most $2m$ new pairs of $x$- and $y$- coordinates, so there are at most pairs $(x,y)$ for which we have encountered a point of the form $(x,y,\star)$. 
    The worst case is that all these points that the $i$th point cannot take satisfy the first condition, but that still leaves at least \[\frac{k^2}{800 \log(k)} - 4km\] points which satisfy the first two conditions.
    Given that the $i$th point satisfies the second condition, the winner in the $z$-direction is $(x,y,2)$ with probability $1/k$. 
    Hence, the probability that the $i$th point is good is at least
    \[\frac{\frac{k^2}{800 \log(k)} - 4km}{k^3} = \Theta(1/(k\log(k)),\]
    and the probability that none of the $m$ points are good is \[1 - 
\Omega \left(\frac{1}{\sqrt{k} \log(k)} \right).\]
    
    The probability that a particular point in $A_4$ does not lead to a good point is at most 
    \begin{align*}
        \frac{4}{5} + \frac{1}{5}\left(1 - \Omega\left(\frac{1}{\sqrt{k} \log(k)} \right) \right) &= 1 - \Omega\left(\frac{1}{\sqrt{k} \log(k)} \right).
    \end{align*}
    Hence, for sufficiently large $k$, the probability that none of the points in $A_4$ lead to a good point is at most
    \[\exp \left( - \Omega\left(\frac{\sqrt{k}}{\log(k)} \right) \right),\]
    as required.
\end{proof}

We now consider the case where $n \geq 4$. To take advantage of the extra dimensions we will consider lots of sequences of good cycles $C_1$, $C_1'$, $C_2'$, $C_2$ where adjacent cycles in the sequence are in slices which differ in a single coordinate. By the above, the probability that there is a path between each pair of adjacent good cycles is $1 - \exp(-\Omega(\sqrt{k}/\log(k)))$, and we will take $n$ independent sequences to boost the probability of getting from $C_1$ and $C_2$.

\begin{lemma}
\label{lem:cycle_to_cycle_single_change}
    There is a constant $c > 0$ such that, for all large enough $k$ and all $n \geq 3$,  with probability at least
     \[1 - \exp( - c n \sqrt{k}/\log(k))\]
     there is a path from $C_1$ to $C_2$ whenever $C_1$ and $C_2$ are good cycles in slices which differ in a single coordinate.   
\end{lemma}

\begin{proof}
    First, let us handle the case where $n  = 3$. 
    Each vertex can be in the basin of at most 2 cycles, and so there are at most $1600 \log(k)$ good cycles in each slice.
    There are $k$ slices, and so there are at most $(1600 k \log(k))^2$ ordered pairs of good cycles.
    The probability that there is not a path between any particular pair of them is $\exp( - \Omega(\sqrt{k}/\log(k)))$ by Lemma~\ref{lem:cycle_to_cycle_3}, and the result follows by the union bound.

    We now assume that $n \geq 4$.
    Suppose $C_1$ and $C_2$ are good cycles in the slices $\mathbf{a}_1$ and $\mathbf{a}_2$ which differ in only one coordinate, and reveal the winners in these slices.
    By relabelling the points as necessary, we can assume $\mathbf{a}_1 = (1, \dots, 1)$ and $\mathbf{a}_2 = (2, 1, \dots, 1)$.
    For each $2 \leq i \leq n - 2$ and odd $1 \leq z \leq k -1$, let $\mathbf{a}_j + z \mathbf{e_i}$ be the vector which differs from $\mathbf{a}_j$ by taking the value $z+1$ in the $i$th position.
    The slices $S(\mathbf{a}_1 + z \mathbf{e_i})$ and $S(\mathbf{a}_2 + z \mathbf{e_i})$ each independently contain a good cycle with probability at least $1/100$ by Lemma~\ref{lem:good_cycle}.
    Reveal the winners in the slices $S(\mathbf{a}_j + z \mathbf{e_i})$ for odd $z$ and suppose that there is a $z$ such that both $S(\mathbf{a}_1 + z \mathbf{e_i})$ and $S(\mathbf{a}_2 + z \mathbf{e_i})$ contain good cycles, which we denote $C_1'$ and $C_2'$ respectively.
    Note that , since there are $\lfloor{k/2}\rfloor$ choices for $z$ for each choice of $i$, the probability that there is such a $z$ is at least 
    \[1 - \left( \frac{1}{100^2} \right)^{\lfloor{k/2}\rfloor} = 1 - \exp(-\Omega(k)).\]

    Consider the subgraph of $\Lnk$ consisting of vertices $\mathbf{v} = (v_1, \dots, v_n)$ where $v_j = 1$ for all $j \not \in \{1,2,i + 2\}$, and note that this is a copy of $\vv{L}(3,k)$.
    By relabelling as necessary we can assume that $z = 1$.
    Hence, we can apply Lemma~\ref{lem:cycle_to_cycle_3} to find that there is a path from $C_1$ to $C_1'$ in this subgraph with probability $1 - \exp(-\Omega(\sqrt{k}/\log(k)))$.
    We note that it does not matter that we have revealed the winners in slices with an even $i$th coordinate as the path does not use any vertex for which the $i$th coordinate is even and at least three.
    Similarly, by considering the subgraph of $\Lnk$ induced by vertices $v = (v_1, \dots, v_n)$ with $v_j = 1$ for $j \not \in \{1,2, 3, i + 2\}$ and $v_{i+2} = z + 1$, there is a path from $C_1'$ to $C_2'$ with probability $1 - \exp(-\Omega(\sqrt{k}/\log(k)))$.
    Likewise, there is a path from $C_2'$ to $C_2$ with probability $1 - \exp(-\Omega(\sqrt{k}/\log(k)))$.
    If all of these paths exist, then there is clearly a path from $C_1$ to $C_2$, which we will call an \emph{$i$-path}.
    After revealing the winners in $S(\mathbf{a}_1)$ and $S(\mathbf{a}_2)$, the existence of an $i$-path is independent of whether there exists a $j$-path for $j \neq i$, so the probability there is a path from $C_1$ to $C_2$ is at least
    \[1 - \exp(- c n \sqrt{k}/\log(k))\]
    for some small constant $c > 0$ and $k$ large enough.
    There are $(n-2)(k-1)k^{n-2}$ (ordered) pairs of slices which differ in exactly one coordinate, and so there are at most $(1600 \log(k))^2 n k^{n-1}$ pairs of good cycles in slices which differ in exactly one coordinate.
    Using the union bound, the probability that there are two good cycles which do not have a path between them is $\exp( - c' n  \sqrt{k}/\log(k))$ for large enough $k$ and some small constant $c' > 0$.
\end{proof}

The lemma above shows that with high probability, there is a path from every good cycle in $S(\mathbf{u})$ to every good cycle in $S(\mathbf{v})$ whenever $\mathbf{u}\mathbf{v}$ is an edge in $H(n-2,k)$. 
Given a good cycle in $S(\mathbf{u})$ and a good cycle in $S(\mathbf{v})$ (where $\mathbf{u} \neq \mathbf{v}$), we know there must be a path between them if there is some path $\mathbf{u} = \mathbf{v_1}, \dots, \mathbf{v_r} = \mathbf{v}$ in $H(n-2, k)$ such that $S(\mathbf{v_i})$ contains a good cycle for every $i$. 
If we say a vertex of $H(n-2, k)$ is \emph{open} if the corresponding slice in $\Lnk$ contains a good cycle, then we are interested in the event that the open vertices form a connected component. 
Since each vertex is open independently with constant probability, this occurs with high probability, as shown by the following lemma.

\begin{lemma}\label{lem:perc}
    Suppose that $n = n(k) \geq 1$ and let $A \subseteq [k]^n$ be a random subset of $[k]^n$ where each element is added independently with probability $p$. Then, for large enough $k$, the subgraph $H(n,k)[A]$ is connected with probability at least $1 - \exp(-cnk)$ for some constant $c > 0$ depending only on $p$.
\end{lemma}

\begin{proof}
    We will show that two events occur with appropriately high probability: that every pair of open vertices at distance at most three in $H(n,k)$ are joined by a path of open vertices, and that every vertex is adjacent to an open vertex. Together these two events are enough to guarantee that $H(n,k)[A]$ is connected. Indeed, fix two vertices $\mathbf{u}, \mathbf{v} \in A$ and consider any path in $H(n,k)$ between them, say $\mathbf{u} = \mathbf{p_0}, \mathbf{p_1}, \dots, \mathbf{p_t} = \mathbf{v}$. 
    For each vertex $\mathbf{p_i}$ with $2 \leq i \leq t-2$, we can pick an open vertex $q_i \in A$ which is adjacent to $p_i$.
    Each of the vertices $\mathbf{p_2}, \dots, \mathbf{p_{t-2}}$ is adjacent to a vertex in $A$, say $\mathbf{p_i}$ is adjacent to the open vertex $\mathbf{q_i}$. Then $\mathbf{u}$ and $\mathbf{q_2}$ are both open and at distance at most three, so are connected by a path of open vertices. Likewise, there is an open path between $\mathbf{v}$ and $\mathbf{q_{t-2}}$. The vertices $\mathbf{q_i}$ and $\mathbf{q_{i+1}}$ are open and at distance at most three, so there also is an open path from $\mathbf{q_i}$ to $\mathbf{q_{i+1}}$. Stitching all these paths together gives an open path from $\mathbf{u}$ to $\mathbf{v}$ as required.

    We now prove that the two events occur with appropriately high probability.
    To prove the claim for the first event, we will need lots of vertex disjoint paths between any two vertices at distance at most three. 
    Let $\mathbf{u}$ and $\mathbf{v}$ be distinct vertices in $A$ at distance at most three in $H(n,k)$, and let us assume that they are actually at Hamming distance exactly three; the other cases are similar. 
    For ease of notation, we consider the copy of $H(n,k)$ on $\{0,1\dots, k-1\}^n$ instead of $[k]^n$.
    By relabelling as necessary, we can assume that $\mathbf{u} = (0, \dots, 0)$ and $\mathbf{v} = (1, 1, 1, 0, \dots, 0)$. Let us associate vertices with points in $\mathbb{R}^n$ in the obvious way, and let $\mathbf{e_1}, \dots, \mathbf{e_n}$ be the standard basis.
    For each $1 \leq \alpha \leq (k-2)/2$ and $j \in [n]$, consider the path 
    \[\mathbf{0}, 2\alpha \mathbf{e_j}, \mathbf{e_1} + 2\alpha \mathbf{e_j}, \mathbf{e_1} + \mathbf{e_2} + 2\alpha \mathbf{e_j}, \mathbf{e_1} + \mathbf{e_2} + \mathbf{e_3} + 2\alpha \mathbf{e_j}, \mathbf{e_1} + \mathbf{e_2} +  \mathbf{e_3}.\]
    These paths are vertex disjoint except for the endpoints (which we have assumed are in $A$).
    Hence, the probability that all of the vertices on this path are in $A$ is $p^4$, and the probability that no path is open for any $\alpha$ and $j$ is at most $(1 - p^4)^{(k-3)n/2}$. 
    There are only $k^{2n}$ possible pairs of open vertices and we can take a union bound over them to find that the probability the first event holds is at least 
    \[1 - \left( k^2 (1-p^4)^{(k-3)/2}\right)^n.\]

    The second claim is even easier.
    Each vertex has $n(k-1)$ neighbours, each of which is independently in $A$ with probability $p$. 
    Hence, the probability that none of them are in $A$ is $(1-p)^{n(k-1)}$.
    Taking a union bound over the $k^n$ vertices in $H(n,k)$, the probability of the second event is at least
    \[1 - \left( k (1-p)^{k-1}\right)^n.\]
    
\end{proof}

We are now in a position to prove the main result of this section, Lemma~\ref{lem:good-cycles-same-component}, which we restate here for convenience.

\goodCyclesSameComponent*

\begin{proof}
    Let $A$ be the set of $\mathbf{v} \in [k]^{n-2}$ for which the slice $S(\mathbf{v})$ contains a good cycle, and suppose that $H(n-2, k)[A]$ is connected and has at least two vertices.
    Suppose also that there is a path in $\Lnk$ between every pair of good cycles in slices which differ in a single coordinate.
    Then every good cycle is in the same strongly connected component.
    Indeed, suppose there are two good cycles $C_1$ and $C_2$ in different slices $S(\mathbf{a_1})$ and $S(\mathbf{a_2})$. 
    Then $\mathbf{a_1}, \mathbf{a_2} \in [A]$ and there is a path $\mathbf{a_1} = \mathbf{v_1}, \dots, \mathbf{v_m} = \mathbf{a_2}$ in $H(n-2, k)[A]$ from $\mathbf{a_1}$ to $\mathbf{a_2}$.
    Each of these vertices corresponds to a slice which contains a good cycle, and it is possible to get from any good cycle in $S(\mathbf{v_i})$ to any good cycle in $S(\mathbf{v_{i+1}})$ as these slices differ in exactly one coordinate. 
    Chaining these together, there is a path in $\Lnk$ from $C_1$ to $C_2$.
    For two good cycles $C_1$ and $C_2$ in the same slice, pick any good cycle $C_3$ in a different slice (which is possible as $|A| \geq 2$) and note that $C_1$ and $C_2$ are both in the same strongly connected component as $C_3$.

    It just remains to check that the events above occur with sufficiently high probability. The probability that $A$ has less than two vertices is at most
    \[\left(\frac{99}{100} \right)^{k^{n-2}} + \frac{k^{n-2}}{100} \left(\frac{99}{100} \right)^{k^{n-2} - 1} \leq \exp(- c' k^{n-2} )\]
    for some small $c' >0 $ and large enough $k$. The other events happen with sufficiently high probability by \Cref{lem:perc,lem:cycle_to_cycle_single_change}.        
\end{proof}

\section{All non-sinks can reach a good cycle}

The main aim of this section is prove Lemma~\ref{lem:reaching_cycle}, which we restate below for convenience.

\reachingCycle*

\begin{proof}[Proof of Lemma~\ref{lem:reaching_cycle}]
    The event that there is a non-sink which cannot reach a good cycle can be equivalently stated as the event that there exist distinct vertices $x,y\in[k]^n$ sharing a line, with the property that $y$ is the winner in their common line and $x$ cannot reach a good cycle. We will union bound over such pairs $(x,y)$, of which there are $n(k-1)k^n$. For each fixed $(x,y)$, the probability that $y$ is the winner in their common line is $1/k$, and we will bound the probability that $x$ cannot reach a good cycle given that $y$ wins the line. 

    In pursuit of this, we begin by considering the following exploration process in $\Lnk$ conditioned on $y$ winning its common line with $x$. At each stage in the process we will have a set $A$ of \emph{active} vertices and a set $P$ of \emph{processed} vertices. Throughout the exploration, it will always be the case that for each active vertex $z\in A$, the winner of exactly one of the lines containing $z$ has been revealed.
    Initiate the process with $A=\{x,y\}$ and $P=\emptyset$. At each step, pick an arbitrary active vertex $z\in A$ (provided one exists). Delete $z$ from $A$ and add it to $P$, then reveal the winners of the $n-1$ lines containing $z$ in which the winner is not currently known. For each such winner $w$ in turn, add $w$ to $A$ unless it shares a line with any vertex in $(A\cup P)\setminus\{z\}$. Any vertex that we add to $A$ at this stage will be called an \emph{offspring} of $z$. Terminate the process at the first point where there are no active vertices at the start of a step, $\abs{A}\geq k^{1/6}$ or $\abs{A \cup P} \geq 4k^{1/6}$.
    In these latter two cases, we may stop part way through adding a collection of winners to $A$.

    Note that, during this process, every line in which we know the winner contains an active or processed vertex, so the condition on adding vertices to $A$ ensures that for each $z\in A$ there are indeed $n-1$ lines containing $z$ whose winner is unknown. Moreover, every vertex on such a line is equally likely to be the winner of the line.

    We first show that the probability that the process ends with $|A \cup P| \geq 4k^{1/6}$ is small.
    \begin{claim}
       There is a constant $c > 0$ such that the probability that the exploration process ends with $|A \cup P| \geq 4k^{1/6}$ is at most $k^{-c n k^{1/6}}$ for all $n \geq 3$ and $k$ large enough. 
    \end{claim}
    \begin{proof}
        Suppose that the exploration process has stopped with $\abs{A\cup P} \geq 4k^{1/6}$. As we stop immediately when this occurs, we have $\abs{A\cup P} = \lceil{4k^{1/6}}\rceil$.
        By Lemma~\ref{lem:trees}, the number of realisations of this exploration process is at most $(enk)^{\lceil{4k^{1/6}}\rceil-1}$.

        Since $|A| \leq \lceil k^{1/6}\rceil$, there are at least $3k^{1/6} - 1$ vertices in $P$ and all but (possibly) one of them has been \emph{fully} processed (since the early termination of the exploration interrupts the processing of at most one vertex).
        As we have terminated the exploration as soon as $\abs{A\cup P} = \lceil{4k^{1/6}}\rceil$, these fully explored vertices have certainly had at most $4k^{1/6}$ offspring between them. 
        That means the number of lines for which we have revealed the winner and it was in a line with a vertex in $(A \cup P) \setminus \{z\}$ is at least
        \[(3k^{1/6} - 2)(n-1) - 4k^{1/6} \geq 2(n-1)k^{1/6} - 3k^{1/6}\]
        for $k \geq 2^{12}$.

        Each time we reveal a winner, the chance that it is any particular vertex on the line is $1/k$ and the chance that it shares a line with a vertex in $(A \cup P) \setminus \{z\}$ is at most $|A \cup P|/k$.
        Each line's revelation is independent so the probability that the exploration process terminates with a particular realisation is at most
        \begin{align*}
            \left(\frac{1}{k}\right)^{\lceil{4k^{1/6}}\rceil - 2} \left(\frac{4k^{1/6}}{k}\right)^{(2n-5)k^{1/6}}.
        \end{align*}
        Hence, taking the union bound over the realisations, the probability that $| A \cup P|$ reaches $\lceil4k^{1/6}\rceil$ is at most
        \[(enk)^{\lceil{4k^{1/6}}\rceil - 1} \left(\frac{1}{k}\right)^{\lceil{4k^{1/6}}\rceil - 2} \left(4k^{-5/6}\right)^{(2n-5)k^{1/6}} \leq k \left( C k^{-5/6}\right)^{(2n-5)k^{1/6}}\]
        for any constant $C > 0$ such that $4 e^4 n^4 \leq C^{2n-5}$ for all $n \geq 3$.
    \end{proof}

    We now bound the probability that the process ends because there are not enough active vertices.

    \begin{claim}
        There is a constant $c > 0$ such that the probability that the process terminates with $|A| = 0$ is at most $(ck)^{-5(n-1)/3}$ for all $n \geq 3$ and large enough $k$.
    \end{claim}
    \begin{proof}
        At every point in the process (except possibly the point at which it terminates) we have $\abs{A \cup P} < 4k^{1/6}$.
        It follows that the probability that the process terminates due to a lack of active vertices is less than the probability that two independent Galton--Watson processes with offspring distributions $\Bin(n-1,1-4k^{-5/6})$ go extinct. By Lemma~\ref{lem:GW_bin}, if $k$ is sufficiently large then this probability is at most
        \[
             4 (4 k^{-5/6})^{2(n-1)}. \qedhere
        \]
    \end{proof}

    Condition on the event that the exploration process has terminated with $|A| = \lceil k^{1/6}\rceil$, and let us first assume that $n \geq 4$.
    We now have many vertices which can be reached from $x$ and from which we have not yet explored. We will explore from these vertices to reach many different slices as follows.
     Initiate $B=\emptyset$. For each $\mathbf{a} \in A$ in turn, reveal for each $i\in \{3,\dots,n\}$ the winner of the line through $a$ in dimension $i$ (provided this has not already been revealed).
     If this winner shares no \emph{slice} with any vertex in $A\cup B\cup P$, then add it to $B$. Note that for each $\mathbf{a}\in A$, we reveal the winner of at least $n-3$ lines through $\mathbf{a}$.

    \begin{claim}
        At each stage, the probability that the winner of the line $l$ through $\mathbf{a}$ in dimension $i \geq 3$ shares no slice with any vertex in $A\cup B\cup P$ is at least $1-(\lceil{4k^{1/6}}\rceil+\lceil{k^{1/6}}\rceil)/k$.
    \end{claim}
    \begin{proof}
        Let $\cS$ denote the set of slices intersecting $l$. First, since $\abs{A\cup P}\leq \lceil{4k^{1/6}}\rceil$, at most $\lceil{4k^{1/6}}\rceil$ slices in $\cS$ contain a vertex in $A\cup P$. It remains to show that at most $\lceil{k^{1/6}}\rceil$ slices in $\cS$ contain a vertex in $B$. Suppose this is not the case. Then there exists $\mathbf{a}'\in A$ with the property that at least two slices in $\cS$ contain vertices which are winners of a line through $\mathbf{a}'$ in dimensions greater than 2 (and are not in $A$). 
        It is not difficult to see that this is impossible.
    \end{proof}

    At the end of the process, the vertices in $B$ are all in different slices, and we have revealed nothing about the edges within these slices. It thus follows from Lemma~\ref{lem:good_cycle} that each vertex of $B$ is in the basin of a good cycle in its slice independently with probability $c/\log(k)$ for some small $c$, provided $k$ is large enough. Thus, the probability that no vertex in $A$ can reach a good cycle in $\Lnk$ is at most
    \[
    \left(1-\frac{c}{\log(k))}\left(1-\frac{5k^{1/6} + 2}{k}\right)\right)^{k^{1/6}(n-3)} \leq e^{-\frac{ck^{1/6}(n-3)}{2\log(k)}}
    \]
    for large enough $k$.

    Clearly, if $n=3$, this argument is not strong enough. The problematic case is when the exploration has revealed that most of the active vertices are the winner in their line in the third dimension.
    To avoid this, we will explore once from each vertex in either the first or second dimension before exploring in the third dimension.
    That is, from each $\mathbf{a}\in A$ reveal the winner of the line through $\mathbf{a}$ in the first dimension (or the second dimension if the winner of the line in the first dimension is already known), and add the winner to a new set $A'$ if it does not share a line with vertices in $A\cup P$ or existing vertices in $A'$. 
    We then emulate the process above, exploring in the third dimension from $A'$ and adding the vertices reached to $B$ if they share no slice with any vertex in $A\cup A'\cup B\cup P$. Adapting the above, we see that the probability that we reach no vertex in a good cycle in this way is at most
    \[
    \left(1 - \frac{c}{\log(k)} \left(1 - \frac{5 k^{1/6} + 2}{k} \right)\left(1 - \frac{6 k^{1/6} + 3}{k} \right)\right)^{k^{1/6}} \leq e^{-\frac{c k^{1/6}}{2 \log(k)}}
    \]
    for some small constant $c$ and large enough $k$.

    We've shown that, for a fixed pair $(x,y)$, the probability that $y$ wins its common line with $x$ and yet cannot reach a good cycle is at most
    \[k^{-c n k^{1/6}} + (ck)^{-5(n-1)/3} + e^{-\frac{ck^{1/6} n}{2\log(k)}}\]
    for some small constant $c > 0$. The result now follows by taking the union bound over the $n k^n$ possible pairs $(x,y)$, noting that $5(n-1)/3 > n$ for $n > 5/2$.
\end{proof}

\section{All good sinks can be reached from some good cycle}\label{sec:goodsinks}

The main aim of this section is to prove \Cref{lem:cycle_to_sink}. We do this by exploring backwards from the sink and showing that the probability that this exploration process grows to $k^{\eps}$ lines but not $k^{2/3}$ lines is negligible. 
Once we have lots of (active) lines in the backwards exploration process, we explore backwards in the third dimension to find around $k^{2/3}$ lines in dimension 3 which can reach the sink.
Heuristically, we expect each of these lines to intersect a good cycle in a given slice with probability around $1/k^{3/2}$, and so we expect these lines to together intersect around $k^{1/6}$ good cycles and the probability that there is no intersection is suitably small.
Although, the total lines ever explored in our exploration is relatively small ($k^{3/4}$), we do reveal some information about the edges within all of the slices, and we will have to show that this does not meaningfully affect the chance that we encounter a good cycle in the slice in future.
We do this by defining ``pathological'' cycles and showing that they are rare.

\subsection{Pathological cycles are rare}

The vertex set of any cycle contained within a slice of $\Lnk$ can be partitioned into two subsets of equal size each consisting of `every other' vertex in the cycle, so that no two vertices in the same half of the partition share a line with one another. In the absence of a clean canonical way to distinguish the two halves of this partition, for each cycle we will arbitrarily refer to one half as the \emph{odd} vertices and the other as the \emph{even} vertices.

\begin{definition}
    A cycle in a slice of $\Lnk$ is \emph{pathological} if 
    \begin{itemize}[nolistsep]
        \item it has length at least $\sqrt{k}$, and
        \item at least half its even vertices and at least half its odd vertices have first or second coordinate taken from $[k^{3/4}]$.
    \end{itemize}
\end{definition}

\begin{lemma}\label{lem:path_cycles}
    The probability that a given slice contains a pathological cycle in $\Lnk$ is at most $k^{-\Omega(\sqrt{k})}$.
\end{lemma}
\begin{proof}
    Identify the vertex set of the slice with $[k]^2$ in the natural way. Since each line can contain at most one even vertex of any given cycle, every pathological cycle has length at most $8k^{3/4}$. Fix $(x_1,y_1)\in[k]^2$ and an integer $\ell$ with $\sqrt{k}\leq 2\ell\leq 8k^{3/4}$. We will upper bound the probability that starting at $(x_1,y_1)$ and repeatedly exploring horizontally then vertically yields a cycle of length $2\ell$ containing $(x_1,y_1)$ in which at least half the vertices of the same parity as $(x_1,y_1)$ have a coordinate in $[k^{3/4}]$.

    Let us assume that $(x_1, y_1)$ is the winner of its line in the $y$-direction as otherwise the exploration cannot possibly yield the desired cycle.
    We can sample the exploration process as follows. Generate a sequence of random points $(x_2, y_2), (x_3, y_3), \dots$ by picking points independently and uniformly at random from all $k^2$ options. The exploration process then follows the sequence 
    \[  (x_1, y_1),
        (x_2, y_1),
        (x_2, y_2),
        (x_3, y_2), \dots \]
    until an $x$- or $y$-coordinate is repeated for the first time, at which point the exploration process stops. We remark that if the exploration finds a cycle of length $2\ell$ containing $(x_1, y_1)$, then the first repeat must be that $x_{\ell + 1} =  x_1$, but we will not need to use this.
    Instead, note that if the exploration yields the desired cycle, then at least $\ell/2 - 1$ of the pairs $(x_2, y_2), \dots, (x_\ell, y_\ell)$ must have at least one coordinate in $[k^{3/4}]$. The probability of this event is at most 
    $2^{\ell-1}(2k^{-1/4})^{\ell/2-1}\leq k^{-\ell/9}\leq k^{-\sqrt{k}/18}$ for large enough $k$. Taking a union bound over the different options for $(x_1,y_1)$ and $\ell$ completes the proof. 

\end{proof}

We now use pathological cycles to prove the following lemma, which shows that revealing that a small subset of the vertices are sources in a slice, does not drastically reduce the probability that there is a good cycle intersecting a given set of vertices.

\begin{lemma}\label{lem:slicewise}
    Let $A\subseteq [k^{3/4}]^2$ have size at most $k^{3/4}$. Let $U\subseteq (k^{3/4},k]^2$ with $\abs{U}\leq  k^{2/3}$. Condition $\vv{L}(2,k)$ on every vertex in $A$ being a source. Then the probability that some vertex in $U$ is in a good cycle in $\vv{L}(2,k)$ is at least $\Omega(|U|/k^{3/2})$.
\end{lemma}
\begin{proof}
    Let $F$ be the event that no vertex in $U$ is in a good cycle in $\vv{L}(2,k)$, and let $D$ be the event that every vertex in $A$ is a source, so we are interested in an upper bound for the probability of $F$ given $D$.
    Note that it is straightforward to see that the event $D$ occurs with high probability. Indeed, the event $D$ occurs if and only if, for every line $l$, the winner of $l$ is not in $A$. Hence, we find
    \[\mathbb{P}(D) = \prod_{l} \left( \frac{k - |A \cap l|}{k} \right) \leq \prod_{l} e^{-|A \cap l|/{k}} = e^{-2k^{3/4}/k}. \]
    The first step in the proof is to show that we can consider the case where there is a good cycle and no pathological cycles.
    To this end, let $E$ be the event that every vertex in $A$ is a source and $\vv{L}(2,k)$ contains a good cycle but no pathological cycles.

    By Lemma~\ref{lem:path_cycles} the probability that $\vv{L}(2,k)$ contains a pathological cycle is at most $k^{-\Omega(\sqrt{k})}$ and by Lemma~\ref{lem:good_cycle} the probability that it contains a good cycle is at least $1/100$.     
    Since $E \subseteq D$, the probability $\Prob(E \mid D)$ is at least $\Prob(E)$, and this is at least $1/200$ for sufficiently large $k$.
    Hence,
    \[\Prob(F \mid D) \leq \Prob(F \mid E) \Prob(E \mid D) + \Prob(E^c \mid D) \leq \frac{1}{200}\Prob(F \mid E) + \frac{199}{200}.\]

    If $E$ occurs, then at least half the even vertices in every good cycle in $\vv{L}(2,k)$ are in $(k^{3/4},k]^2$. Select a good cycle in $\vv{L}(2,k)$ uniformly at random, and then let $(x_1,y_1),\dots,(x_t,y_t)\in (k^{3/4},k]^2$ be $t\coloneqq \lceil\sqrt{k}/4\rceil$ uniformly random distinct even vertices on the cycle from $(k^{3/4}, k]^2$. Observe that, even having conditioned on $E$, these vertices have the same distribution as uniformly drawing $t$ points (without repetition) from $(k^{3/4},k]^2$ such that no two share a line. Let $Z=\{(x_i,y_i):i\in[t]\}$, then it is sufficient to show that $\Prob(U\cap Z=\emptyset)\leq 1-\Omega(\abs{U}/k^{3/2})$.

    \begin{claim}
    For $k$ large enough, we have that
        \[
        \Prob\big((x_i,y_i)\in U \mid  (x_j,y_j)\not\in U ~ \forall j<i\big)
        \geq \frac{\abs{U}}{k^{2}} \left(1-  \frac{2}{k^{5/6}}\right)
        \]
        for all $i \in [t]$.
    \end{claim}
    \begin{proof}
        We have
        \[
        \begin{split}
            \Prob\big((x_i,y_i)\in U \mid (x_j,y_j)\not\in U ~ \forall j<i\big)
            & \geq \Prob\big((x_i,y_i)\in U  \wedge  (x_j,y_j)\not\in U ~ \forall j<i\big)
        \end{split}
        \]
    The probability that $(x_i,y_i)\in U$ is simply $|U|/(k-\lfloor k^{3/4} \rfloor)^2 \geq \abs{U}/k^2$. Given $(x_i,y_i)$ and $(x_1,y_1),\dots,(x_{j-1},y_{j-1})$ for some $j<i$, $(x_j,y_j)$ is uniformly distributed amongst the points in $[k^{3/4},k)^2$ which do not share a line with any of these given points. Hence, the probability that $(x_j,y_j)\not\in U$ given these points is at least
    \[
    \frac{(k-\lfloor k^{3/4}\rfloor-j)^2-\abs{U}}{(k-\lfloor k^{3/4}\rfloor-j)^2}
    = 1-\frac{\abs{U}}{(k-\lfloor k^{3/4}\rfloor-j)^2}
    \geq 1-\frac{k^{2/3}}{(k-2k^{3/4})^2}.
    \]
    Hence, applying Bernoulli's inequality, we obtain (for large enough $k$)
    \[
    \begin{split}
        \Prob\big((x_i,y_i)\in U \mid  (x_j,y_j)\not\in U ~ \forall j<i\big)
        & \geq\frac{\abs{U}}{k^2}\left( 1-\frac{k^{2/3}}{(k-2k^{3/4})^2}\right)^i \\
        & \geq \frac{\abs{U}}{k^2}\left(1 - \frac{2i}{k^{4/3}}\right) \\
        & \geq \frac{\abs{U}}{k^{2}} \left(1-  \frac{2}{k^{5/6}}\right). \qedhere
    \end{split}
    \]
    \end{proof}
    It follows that, for large $k$,
    \[
    \begin{split}
    \Prob(U\cap Z=\emptyset)
    & = \prod_{i=1}^{t}{\Big(1-\Prob\big((x_i,y_i)\in U \mid (x_j,y_j)\not\in U ~ \forall j<i\big)\Big)} \\
    & \leq \left[1-\frac{\abs{U}}{k^2}  \left(1-  \frac{2}{k^{5/6}}\right) \right]^{\sqrt{k}/4} \\
    & \leq \exp\left(-\frac{\abs{U}}{5k^{3/2}} \right) \\
    & \leq 1-\frac{\abs{U}}{10k^{3/2}},
    \end{split}
    \]
    as required.
\end{proof}

\subsection{Proof of Lemma \ref{lem:cycle_to_sink}}

Consider the following exploration process on $\Lnk$. At each stage in the process there will be a set $A$ of \emph{active} lines and a set $P$ of \emph{processed} lines. Initialise $P = \emptyset$ and initialise $A$ by picking a vertex $x \in [k]^n$, revealing in turn which (if any) of its lines it is the winner in, and adding any such lines to $A$.
At each subsequent step, arbitrarily pick an active line $l\in A$ for processing. Add $l$ to $P$ and delete it from $A$. Now for each line $l'$ intersecting $l$ in which we have not already revealed the winner, reveal whether or not the winner of $l'$ is in $l$. If it is, then add $l'$ to $A$. The process terminates as soon as any of the following are true:
\begin{enumerate}[label = (\roman*)]
    \item there are no active lines at the start of a step, or
    \item $\abs{A}\geq k^{2/3}$, or
    \item $\abs{A \cup P} \geq k^{3/4}$.
\end{enumerate}
Note that termination due to conditions (ii) and (iii) could occur in the middle of a step, including during the initialisation step if $n$ is large enough.

We start by recording some simple facts about this exploration.
\begin{enumerate}[label=(\alph*)]
    \item Every vertex on a line in $A\cup P$ can reach $x$ along a directed path in $\Lnk$. Moreover, if the process terminates due to condition (i), then the lines in $A\cup P$ contain exactly the vertices that can reach $x$ along directed paths in $\Lnk$.
    \item During the exploration, if the winner of a given line has been revealed, then that line is in $A\cup P$.\footnote{Note that even if every vertex but one in a line has been revealed \emph{not} to be the winner of that line, we do not consider the winner of that line to have been revealed. However, as soon as the winner of a line is revealed, we consider it to be revealed that the other vertices of that line are not the winner.}
    \item \label{d} It follows that when we begin processing a line $l$, the number of lines intersecting $l$ in which we have not already revealed the winner is at least $(n-1)k-\abs{A\cup P}$.
    \item If it has been revealed that a given vertex is not the winner in a given line, then that vertex is in a line in $A\cup P$ (except in the trivial case where no lines ever are added to $A$).
    \item It follows that when revealing the winner of a line, the probability of success is always at least $1/k$ and at most $1/(k-\abs{A\cup P})$.
\end{enumerate}

We first show that the probability the process terminates because of condition (iii) is sufficiently small.

\begin{lemma}\label{lem:bushy_trees}
    The probability that the process terminates with $\abs{A\cup P}\geq k^{3/4}$ is at most $e^{-\Omega(nk^{3/4})}$.
\end{lemma}
\begin{proof}
    At the point of termination, $\abs{A\cup P}=\lceil k^{3/4}\rceil$. By Lemma~\ref{lem:trees}, the number of possible realisations of this exploration is at most $n(e(n-1)k)^{k^{3/4}}$, where the additional factor of $n$ comes from fixing an initial active line through $x$ as the root of the tree. 
    As the walk has not already terminated, this is the first point for which condition (ii) might be satisfied and we have $\abs{A}\leq \lceil{k^{2/3}}\rceil$, so at least $\lceil k^{3/4} \rceil - \lceil{k^{2/3}} \rceil-1=(1+o(1))k^{3/4}$ lines have been fully processed. By (d) above, during the processing of each of these lines, at least $(n-1)k-\lceil{k^{3/4}}\rceil=(1+o(1))(n-1)k$ lines are checked.
    There were at most $k^{3/4}$ successful line checks during the full exploration, so during the processing of the fully processed lines there were at least
    \[
    \Big((1+o(1))k^{3/4}\Big)\cdot\Big((1+o(1))(n-1)k\Big)-k^{3/4}=(1+o(1))(n-1)k^{7/4}
    \]
    unsuccessful line checks. 
    
    By (f), the probability of a given realisation of this form occurring is at most
    \[
    \left(\frac{1+o(1)}{k}\right)^{k^{3/4}} 
    \left(1-\frac{1}{k}\right)^{(1+o(1))(n-1)k^{7/4}}
    \leq \left(\frac{1+o(1)}{k}\right)^{k^{3/4}}e^{-(1+o(1))(n-1)k^{3/4}}
    \]
    Taking a union bound over the possible realisations shows that the probability that the process ends with $|A \cup P| \geq k^{3/4}$ is at most
    \[
    \begin{split}
    n(e(n-1)k)^{k^{3/4}} & \left(\frac{1+o(1)}{k}\right)^{k^{3/4}}e^{-(1+o(1))(n-1)k^{3/4}} \\
    & = ne^{-(1+o(1))(n-2)k^{3/4}}(n-1)^{k^{3/4}}(1+o(1))^{k^{3/4}} \\
    & = ne^{-(1+o(1))(n-2-\log(n-1))k^{3/4}} \\
    & \leq e^{-nk^{3/4}/10}
    \end{split}
    \]
    for all $n\geq 3$ if $k$ is sufficiently large.
\end{proof}

Next, we bound the probability that the root vertex $x$ of the exploration is a sink and the process grows to $\abs{A\cup P}\geq k^{\eps}$, yet the exploration terminates due to a lack of active lines. 

\begin{lemma}\label{lem:good_sink_explodes}
    The probability that $x$ is a sink and $\abs{A\cup P}$ grows to size at least $k^\eps$ yet the process terminates due to condition (i) is at most $e^{-\Omega(k^\eps)}/k^n$.
\end{lemma}
\begin{proof}
    We will take a union bound over the events that $x$ is a sink and the exploration terminates due to condition (i) being met when $\abs{A\cup P}=\abs{P}=i$ for each $k^{\eps}\leq i< k^{3/4}$.
    Fix such an $i$. Clearly we may assume that $i\geq n$ as we are assuming that $x$ is a sink. We start with the following claim, which can be proved via a straightforward adaptation of the proof of Lemma~\ref{lem:trees} given in~\citet*{mcdiarmid2021component}.
    
    \begin{claim}\label{claim:tree_count_mod}
        The number of realisations of the exploration in which $x$ is a sink and $\abs{A\cup P}=i$ is at most 
        \[
        \binom{(i-1)(n-1)k}{i-n}
        \leq \left(\left(1+\frac{n-1}{i-n}\right)e(n-1)k\right)^{i-n}
        \]
        if $i>n$, and exactly 1 if $i = n$.
    \end{claim}
    
    Each realisation which terminates due to (i) at a point where $\abs{A\cup P}=i$ involves exactly $i$ successful line checks and hence, by (d) above, at least
    \[
    i\big((n-1)k-i\big)-i = (1+o(1))i(n-1)k
    \]
    unsuccessful line checks. By (f) above, the probability of a successful line check is always at least $1/k$ and at most $1/(k-i) \leq (1 + 2i/k)/k$ for large enough $k$. The probability of a given realisation occurring is thus at most
    \[
    \left(\frac{1 + 2i/k}{k}\right)^i\left(1-\frac{1}{k}\right)^{(1+o(1))i(n-1)k}
    \leq \frac{e^{-(1+o(1))i(n-1) - 2i^2/k}}{k^i} = \frac{e^{-(1+o(1))i(n-1) }}{k^i}
    \]
    If $i=n$, then this probability is $e^{-\Omega(k^\eps)}/k^n$ as required. Otherwise, by a union bound, the probability that $x$ is a sink and the exploration terminates due to (i) at a point where $\abs{A\cup P}=i$ is at most
    \begin{align*}
        & \left(\left(1+\frac{n-1}{i-n}\right) e(n-1)k\right)^{i-n} \cdot \frac{e^{-(1+o(1))i(n-1) }}{k^i}\\ 
        & \leq \left(\frac{1}{\left(1+\frac{n-1}{i-n}\right)e(n-1)k}\right)^n \cdot \left[\left(1+\frac{n-1}{i-n}\right)e(n-1) e^{-(1+o(1))(n-1)}\right]^i\\
        &\leq \frac{1}{k^n}\cdot \left[\left(1+\frac{n-1}{i-n}\right)e(n-1) e^{-(1+o(1))(n-1)}\right]^i.
    \end{align*}

    For $n \geq 6$, we have $en(n-1)e^{-(n-1)} \leq 4/5$ and so, using the bound $1+\frac{n-1}{i-n}\leq 1+n-1=n$, we have that if the above is bounded by $(4/5)^{k^{\eps}}/k^n$ for all sufficiently large $k$ (not depending on $n$).
    If $n \in {3,4,5}$, we note that $e(n-1)e^{-(n-1)} \leq 3/4$ and $1 + \frac{n-1}{i-n} \leq 1 + \frac{4}{k^{\eps} - 5}$ so, for large enough $k$, the above is again bounded by $(4/5)^{k^{\eps}}/k^n$.    
    Taking a union bound over $i$ now completes the proof.
\end{proof}

We are now ready to prove Lemma~\ref{lem:cycle_to_sink}.

\begin{proof}[Proof of Lemma~\ref{lem:cycle_to_sink}]
    By a union bound and symmetry, the probability that there exists a good sink which cannot be reached from any good cycle is at most
    \[
    k^n\cdot\Prob(x\textrm{ is a good sink but cannot be reached from a good cycle}),
    \]
    for any fixed $x\in [k]^n$. 
    Note that if $x$ is a good sink and the process terminates as there are no active lines at the start of a step, then $\abs{A \cup P}$ contains all of the lines that can reach $x$ and $\abs{A \cup P} \geq k^{\eps}$.
    Hence, the probability that $x$ is a good sink that cannot be reached from a good cycle is at most the sum of the probabilities that 
    \begin{itemize}
        \item $x$ is a sink and the exploration process from $x$ terminates by condition (i) at a time when $\abs{A\cup P}\geq k^{\eps}$,
        \item the process terminates with $\abs{A\cup P}\geq k^{3/4}$, and
        \item $x$ is a sink and the exploration terminates due to condition (ii), but $x$ cannot be reached from a good cycle.
    \end{itemize}
    By Lemmas~\ref{lem:bushy_trees} and ~\ref{lem:good_sink_explodes}, the first two probabilities are sufficiently small, so we focus on the third.

    Clearly, the probability that $x$ is a sink and the exploration process terminates due to condition (ii) is at most $k^{-n}$. Hence, to prove the theorem it is sufficient to show that, conditioned on this event, the probability that $x$ cannot be reached from a good cycle is at most $e^{-\Omega(k^{\eps})}$. In pursuit of this, start by relabelling $[k]$ separately in each dimension so that each line in $\abs{A\cup P}$ intersects $[k^{3/4}]^n$.

    We will now explore in the first dimension from all lines in $A$, noting that $\abs{A}=\lceil{k^{2/3}}\rceil$. However, we want to leave some slices as unexplored as possible so that we can join up with a good cycle in these later; for this reason we do not explore from any vertices with third coordinate greater than $k/2$. For simplicity, we will in fact not explore from any vertex outside $[k/2]^n$.

    More formally, we perform the following exploration. Let $A'\subseteq A$ be the set of lines in $A$ in a dimension other than the first, and let $T=|A'|$. Let $\cL$ be the set of lines in the first dimension which intersect any line in $A'$ at a vertex in $[k/2]^n\setminus [k^{3/4}]^n$. For each $L\in \cL$, write $t_L$ for the number of incidences between $L$ and a line in $A'$. Note that the total number of incidences between lines in $\cL$ and lines in $A'$ is $\sum_L{t_L}=T(\lfloor k/2 \rfloor-\lfloor k^{3/4} \rfloor)$. 

    The only information that has been revealed about the winner of each line in $\cL$ is that the winner is not any vertex on a line in $P$. Thus, the probability that the winner is on a line in $A'$ is at least $t_L/k$, and this event occurs independently for each $L\in\cL$. Hence, the number of lines in $\cL$ whose winner is in a line in $A'$, denoted $X$, dominates the sum over $L\in \cL$ of independent Bernoulli random variables with parameter $t_L/k$. Applying Lemma~\ref{lem:chernoff}, where here $\mu = \mathbb{E}[X] \geq T(1/2-k^{-1/4}-1/k)$ and $\delta =1/4$, we have (for large~$k$)
    \[
    \Prob(X\leq T/3)\leq \Prob(X\leq 3\mu/4) \leq e^{-T(1/2-k^{-1/4} - 1/k)/32}\leq e^{-T/70}.
    \]
    Thus, either $T\leq 3k^{2/3}/4$, in which case there are at least $k^{2/3}/4$ lines in $A$ in the first dimension, or $T\geq 3k^{2/3}/4$, in which case with failure probability $e^{-\Omega(k^{2/3})}$ the above exploration process yields at least $k^{2/3}/4$ lines in the first dimension which can reach $x$.

    We can therefore condition on having obtained, either directly from $A$ or from this second exploration process, a set $B_1$ of at least $k^{2/3}/4$ lines in the first dimension, all of which can reach $x$, and all of which intersect $[k/2]^n$. We now explore in the second dimension from the lines in $B_1$ in a very similar way: reveal which lines in the second dimension with first coordinate greater than $k^{3/4}$ have their winner in a line in $B_1$. Similarly to above, with failure probability at most $e^{-\Omega(k^{2/3})}$, this yields a set $B_2$ of at least $k^{2/3}/16$ lines in the second dimension, all of which can reach $x$, and all of which have first coordinate greater than $k^{3/4}$ and third coordinate at most $k/2$.

    Finally, we explore in the third dimension from the lines in $B_2$: reveal which lines in the third dimension with second coordinate greater than $k^{3/4}$ have their winner in a line in $B_2$. With failure probability at most $e^{-\Omega(k^{2/3})}$ this yields a set $B_3$ of at least $k^{2/3}/64$ lines in the third dimension, all of which can reach $x$, and all of which have first and second coordinate greater than $k^{3/4}$. If necessary, arbitrarily throw away lines in $B_3$ one at a time until $|B_3| \leq k^{2/3}$.

    We now consider the slices of $\Lnk$ with third coordinate greater than $k/2$. All we have revealed about the restriction of $\Lnk$ to each such slice is that some given collection of at most $k^{3/4}$ vertices contained in $[k^{3/4}]^2$ (a subset of the vertices at which the slice intersects lines in $P$) are sources in the slice. It therefore follows from Lemma~\ref{lem:slicewise} that for each fixed slice $S$ with third coordinate greater than $k/2$ the probability that no vertex in a line in $B_3$ is in a good cycle in $S$ is at most $1-\Omega(u_S/k^{3/2})\leq e^{-\Omega(u_S/k^{3/2})}$, where $u$ is the number of lines in $B_3$ that intersect $S$. Since $|B_3|\geq k^{2/3}/64$ and each line in $B_3$ intersects $k/2$ slices with third coordinate greater than $k/2$, there are a total of at least $k^{5/3}/128$ incidences between these slices and lines in $B_3$. The events that no vertex in a line in $B_3$ is in a good cycle in a given slice $S$ with third coordinate greater than $k/2$ are independent over the different choices for $S$, so the probability that no vertex in a line in $B_3$ is in a good cycle in any such slice is at most
    \[
    e^{-\Omega(k^{5/3}/k^{3/2})}= e^{-\Omega(k^{1/6})},
    \]
    as required.
\end{proof}

\end{document}